\documentclass[11pt]{article}
\usepackage{lmodern}
\usepackage[ruled,vlined]{algorithm2e}
\usepackage{amsfonts}
\usepackage{amsmath}
\usepackage{natbib}
\usepackage{amsthm}
\usepackage{amssymb}
\usepackage{color}
\usepackage{comment}
\usepackage{enumitem}
\usepackage{float}
\usepackage[colorlinks=true,allcolors=blue]{hyperref}
\usepackage{tikz}
\usepackage{nicefrac}
\usepackage{multicol}
\usepackage{stmaryrd}
\usepackage[margin=1in]{geometry}
\usepackage[utf8]{inputenc}

\tikzset{
    vertex/.style={},
    edge/.style={->}
}

\makeatletter
\newtheorem*{rep@theorem}{\rep@title}
\newcommand{\newreptheorem}[2]{%
\newenvironment{rep#1}[1]{%
 \def\rep@title{#2 \ref{##1}}%
 \begin{rep@theorem}}%
 {\end{rep@theorem}}}
\makeatother

\newcommand{\prob}{\mathbb{P}}
\newcommand{\probsymbol}{\mathbf{P}}

\newcommand{\ETR}{\exists\mathbb{R}}
\newcommand{\NP}{\mathsf{NP}}
\newcommand{\PSPACE}{\mathsf{PSPACE}}
\newcommand{\SAT}{\mathsf{SAT}}

\newreptheorem{theorem}{Theorem}
\newtheorem{theorem}{Theorem}
\newtheorem{lemma}{Lemma}[section]
\newtheorem{proposition}[lemma]{Proposition}
\theoremstyle{definition}

\newtheorem{example}[lemma]{Example}
\newtheorem{definition}[lemma]{Definition}

\title{Is Causal Reasoning Harder than Probabilistic Reasoning?}

\author{Milan Mossé, Duligur Ibeling, Thomas Icard}
\date{}

\begin{document}

\maketitle

\begin{abstract} Many tasks in statistical and causal inference can be construed as problems of \emph{entailment} in a suitable formal language. We ask whether those problems are more difficult, from a computational perspective, for \emph{causal} probabilistic languages than for pure probabilistic (or ``associational'') languages. Despite several senses in which causal reasoning is indeed more complex---both expressively and inferentially---we show that causal entailment (or satisfiability) problems can be systematically and robustly reduced to purely probabilistic problems. Thus there is no jump in computational complexity. 
Along the way we answer several open problems concerning the complexity of well known probability logics, in particular demonstrating the $\ETR$-completeness of a polynomial probability calculus, as well as a seemingly much simpler system, the logic of comparative conditional probability. 
\end{abstract}

\section{Motivation and Preview}

There is an uncontroversial sense in which causal reasoning is more difficult than purely  probabilistic or statistical reasoning. The latter seems hard enough: estimating probabilities, predicting future events from past observations, determining statistical significance, adjudicating between statistical hypotheses---these are already formidable tasks, long mired in controversy. No free lunch theorems \citep{MLbook,BELOT2020159} show that strong assumptions are necessary to gain any inductive purchase on such problems, and there is considerable disagreement about what kinds of assumptions are reasonable in different epistemic and practical circumstances \citep{Efron}. Problems of causal inference only seem to make our tasks harder. Inferring causal effects, predicting the outcomes of interventions, determining causal direction, learning a causal model---these problems typically demand statistical reasoning, but they also demand more on the part of the investigator. They may require that we actively interrogate the world through deliberate experimentation rather than passive observation, or that we antecedently accept strong assumptions sufficient to justify the causal conclusions we want to reach, or (very often) both. Indeed, statistical indistinguishability is the norm in causal inference, even with substantive assumptions \citep{spirtes2000causation}. As formalized in the causal hierarchy theorem of \cite{bareinboim2020pearl} (see also \citealt{II21}), it is not only impossible to infer causal information from purely correlational (or ``observational'') data, but also generically impossible to infer counterfactual or explanatory information from purely experimental (or ``interventional'') data. From an inferential perspective, probabilistic information vastly underdetermines causal information.

A feature common to both statistical inference and causal inference is that the most prominent approaches to each can be understood, at least in part, as attempts to turn an inductive problem into a deductive one. This is famously true of frequentist methods in the tradition associated with Neyman and Pearson (see \citealt{Neyman}), but is arguably true of Bayesian approaches as well. As \cite{GelmanShalizi} suggest, ``Statistical models are tools that let us draw inductive inferences on a deductive background,'' rendering statistical inferences  ``deductively guaranteed by probabilistic assumptions'' (p. 27). 
Indeed, one of the benefits of specifying a Bayesian probability model is that it provides an answer to virtually any question about the probability of a hypothesis conditional on data. Given the model and the data, this answer follows as a matter of logic. 

Causal underdetermination is likewise confronted with methods for formulating precise inductive assumptions, sometimes allowing answers to causal questions to be derived by mere calculation. 
\begin{example}[Do-calculus] \label{ex:do} As one prominent example, the \emph{do-calculus} of Pearl and collaborators (see \citealt{pearl1995} and Ch. 3 of \citealt{Pearl2009}) establishes systematic correspondences between qualitative (``graphical'') properties of a causal scenario and certain conditional independence statements involving causal quantities. 
A typical causal quantity of interest is the \emph{(average) causal effect}, e.g., how likely $Y$ is to take on value $y$ given an intervention setting $X$ to $x$. In a formal language (introduced in the sequel as $\mathcal{L}_{\mathrm{\mathrm{causal}}}$), we write this as $\probsymbol([X=x]Y=y)$, or more briefly, $\probsymbol([x]y)$.  

Absent assumptions, it is \emph{never}  possible to infer the value of $\probsymbol([x]y)$ from observational data \citep{bareinboim2020pearl}. 
Suppose, however, that we could assume the causal structure has something like the following shape (known in the literature as the \emph{front door  graph}):  
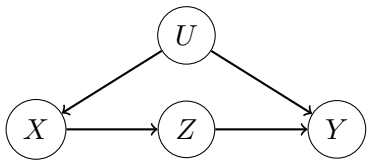
\begin{figure}[h]
\begin{center}
\begin{tikzpicture}
  \node (s0) [draw=black,circle] at (0,0) {$X$};
  
  \node (s2) [draw=black,circle] at (2,0) {$Z$};
  
  \node (s3) [draw=black,circle] at (4,0)  {$Y$};
  
  \node (s4) [draw=black,circle] at (2,1.25) {$U$};
 
  \path (s0) edge[->,thick] (s2);
  
  \path (s2) edge[->,thick] (s3);
  
    \path (s4) edge[->,thick] (s0);
  
  \path (s4) edge[->,thick] (s3);
  
 \end{tikzpicture} 
 \end{center} \vspace{-.2in}
 \end{figure} 
 
 \noindent For a standard example, we might assume that any causal effect of smoking ($X$) on cancer ($Y$) will be mediated by tar deposited in the lungs ($Z$), 
 and moreover that any unknown sources of variation ($U$) on $X$ or on $Y$ (or on both), such as a person's genotype, do not directly influence $Z$. Under these circumstances, the do-calculus licenses several substantive causal assumptions, which may be rendered precisely in $\mathcal{L}_{\mathrm{\mathrm{causal}}}$. Let $\Gamma$ be the set of equality statements below: 
 \begin{multicols}{2}
 \begin{enumerate}[label=(\roman*)]
     \item \label{do-1} $\probsymbol([x]z) = \probsymbol(z|x)$
     \item $\probsymbol([z]x) = \probsymbol(x)$
     \item $\probsymbol([x]y  | [x]z)=\probsymbol([x,z]y)=\probsymbol([z]y)$
     \item \label{do-4} $\probsymbol([z]y  | [z]x)=\probsymbol(y|x,z)$
 \end{enumerate}
 \end{multicols}
\noindent For instance, \ref{do-1} says that the causal effect of $X=x$ on $Z=z$ simply coincides with the conditional probability $\prob(Z=z|X=x)$. Appealing to a combination of laws of probability and distinctively causal laws involving the ``causal-conditional'' statements like $[x]y$, it is possible to show that the following equality is in fact \emph{entailed} by the statements $\Gamma$, that is, by \ref{do-1}-\ref{do-4}:
\begin{equation} \label{eq:do} \probsymbol([x]y) \;\; = \;\; \sum_z \probsymbol(z|x) \sum_{x'} \probsymbol(y|x',z) \probsymbol(x'). \end{equation}
In other words, (\ref{eq:do}) shows that the causal effect of $X=x$ on $Y=y$ can simply be calculated from suitable observational data involving the variables $X,Y,Z$. \end{example}

Methods such as these extend beyond the specific problem of estimating causal effects, to include estimation of \emph{counterfactual} quantities as well. For instance, we may want to determine---from experimental data and background assumptions---the joint probability that an individual would survive if \emph{and only if} they are assigned a certain treatment, a quantity we would write as $\probsymbol([X=1]Y=1 \wedge [X=0]Y=0)$. Inferential techniques similar to those in Example \ref{ex:do} have been employed in such settings, and have even been automated (e.g., \citealt{Duarte}).

More broadly, a number of different approaches to inductive inference, both statistical and causal, can be assimilated to a regiment something like this: \begin{equation} \mbox{Inductive Assumptions } + \mbox{ Data } \models \mbox{ Inferential Conclusion}\label{mainequation}\end{equation} In Example~\ref{ex:do}, $\Gamma$ are the inductive assumptions, the data would be information about $\prob(X,Y,Z)$, and the conclusion would be an estimate of the causal effect of $X=x$ on $Y=y$. In a standard Bayesian analysis, the inductive assumption might be a prior probability model for some latent variables (e.g., parameters for a class of probability measures), while the data would be values of some observable variables, and the conclusion might be the posterior values for the hidden variables, or perhaps posterior predictive values for some yet-to-be-observed variables. A critical job of the statistician or data scientist is to identify suitable inductive assumptions that a relevant party judges reasonable (or, ideally if feasible, which are themselves empirically verifiable) and that are sufficiently strong to license meaningful conclusions from the types of data available. 

From this vantage point our titular question takes on a new significance. Rather than asking about the difficulty of an inference task in terms of the \emph{strength of assumptions} needed to justify the inference, we could instead ask how difficult it is in general, \emph{computationally speaking}, to reason from inductive assumptions (together with data) to an inferential conclusion, in the strong sense of (\ref{mainequation}). 
In other words, we ask how difficult questions like (\ref{mainequation}) could be across different logical languages for describing relevant assumptions, data, and conclusions. 

The contrast of interest in this article is between languages $\mathcal{L}_{\mathrm{prob}}$, suitable for probabilistic reasoning, and languages $\mathcal{L}_{\mathrm{causal}}$, which extend the corresponding probabilistic languages to encompass causal reasoning in addition. In short, $\mathcal{L}_{\mathrm{prob}}$ encompasses ``pure'' probabilistic reasoning about some set of random variables. In $\mathcal{L}_{\mathrm{causal}}$ we also reason about the probabilities of \emph{causal conditionals}, the causal effect $\probsymbol\big([x]y\big)$ being a simple example. Such mixed reasoning is crucial for applications like the do-calculus, where causal conclusions depend on distinctively causal assumptions (such as \ref{do-1}-\ref{do-4} in Example \ref{ex:do}). Some of the emblematic principles of $\mathcal{L}_{\mathrm{causal}}$ reveal a subtle interplay between the probabilistic and causal-conditional components. For example, the following formula states that if causal interventions which set the values of the variable $X$ thereby affect the values taken the variable $Y$, then the converse cannot be true:\begin{equation} \label{eq:rec} \probsymbol\big([x]y \wedge [x']y'\big)>0 \rightarrow \probsymbol\big([y]x \wedge [y']x'\big)=0.
\end{equation} This formula emerges as an instance of a more general scheme in a complete axiomatization of $\mathcal{L}_{\mathrm{causal}}$ (see \citealt{ibeling2020probabilistic}), implying that $X$ and $Y$ cannot each causally affect the other.

In light of the considerable \emph{empirical} (and expressive) gulf between these two kinds of languages, we might expect to see a parallel jump in computational complexity when moving from $\mathcal{L}_{\mathrm{prob}}$ to $\mathcal{L}_{\mathrm{causal}}$. In a certain respect, $\mathcal{L}_{\mathrm{causal}}$ can be seen as a \emph{combination} of logics, embedding one modal system (a conditional logic) inside another (a probability logic),  with non-trivial interactions between the two (such as (\ref{eq:rec})). It is common wisdom that such combinations may in general drive up complexity, in some cases even resulting in undecidability (see, e.g., \citealt{Kurucz}). As a famous example, even seemingly innocuous combinations of modalities for knowledge and time (each independently of low complexity) can lead to $\Pi^1_1$-hardness \citep{HalpernVardi}. 
The present work introduces two main results, which show that this does not happen here: causal reasoning and probabilistic reasoning are, in a precise and robust sense, equally difficult.

The distinction between  $\mathcal{L}_{\mathrm{prob}}$ and $\mathcal{L}_{\mathrm{causal}}$ is orthogonal to another distinction, namely \emph{how much arithmetic} we admit in our formal language of probability over a set of probability terms $\probsymbol(\delta)$. A wide range of probability logics have been studied in the literature, from pure qualitative comparisons between probability terms (e.g., \citealt{Finetti1937}) to richer fragments capable of reasoning about polynomials over such terms (e.g., \citealt{Scott1966AssigningPT}).  For any such choice $\mathcal{L}_{\mathrm{prob}}$ of probabilistic language we can consider the extension $\mathcal{L}_{\mathrm{causal}}$ to allow not only probability terms, but also causal-probability terms like those introduced above. A strength of our analysis is that we provide a complexity-reflecting reduction from $\mathcal{L}_{\mathrm{causal}}$ to $\mathcal{L}_{\mathrm{prob}}$ in a way that is independent of our choice of probabilistic primitives. Thus, across the landscape of probability logics, we see no increase in complexity. Summarizing, our main result states:




\begin{theorem}[Informal]\label{characterization}
Probabilistic reasoning is no harder than causal reasoning. In particular:
\begin{enumerate}
    \item 
    Reasoning about (causal or non-causal) probabilities is as hard as reasoning about sums of (causal or non-causal) probabilities; both are as hard as reasoning about Boolean formulas.
    \item \label{enum-second}
    Reasoning about (causal or non-causal) conditional probabilities is as hard as reasoning about arbitrary polynomials in (causal or non-causal) probabilities; both are as hard as reasoning about arbitrary polynomials in real numbers.
\end{enumerate}
\end{theorem} While the relationship between probabilistic and causal languages is our main focus, it is worth pointing out that some of our results are of interest beyond the connection with causality. In particular, we find that reasoning in the language of conditional comparative probability is precisely as hard as reasoning in the full existential first-order theory of real numbers ($\mathsf{\exists}\mathbb{R}$), thus establishing another notable example of a problem complete for this complexity class. It is also noteworthy that this expressively weak probabilistic language is---from a computational perspective---as complex as the most expressive \emph{causal} languages we consider in the paper (namely, $\mathcal{L}_{\mathrm{causal}}^{\mathrm{poly}}$).


\subsubsection*{Relation to previous work} There is a long line of work on probability logic, including a host of results about complexity \citep{fagin1990logic,Abadi,PL,Speranski}. As just mentioned, our contribution advances this literature. Concerning causal reasoning, there have been a number of complexity studies for various non-probabilistic causal notions \citep{Eiter,Aleksandrowicz}. Most germane to the present study is Halpern's \citeyearpar{Halpern2000} analysis of the satisfiability problem for deterministic reasoning about causal models, which he shows to be $\NP$-complete (the same as propositional logical reasoning).  \cite{Eiter} studied numerous model-checking queries in a probabilistic setting, including the problem of determining the probability of a specific causal query. They show that this problem is complete for the class $\#\mathsf{P}$, the ``counting analogue'' to $\NP$ which also characterizes the problem of determining (approximations for) probabilities of (even very simple) propositional expressions \citep{ROTH1996273}.

Our interest in the present contribution is the complexity of reasoning---viz. testing for satisfiability, validity, or entailment, as portrayed in (\ref{mainequation})---for probabilistic and causal languages. While this angle has not yet been explored thoroughly in the literature, our study is indebted to, and draws upon, much of this previous work. 
Theorem~\ref{characterization} synthesizes as well as greatly extends a heretofore piecemeal line of results \citep{fagin1990logic,ibeling18,ibeling2020probabilistic}. Moreover, the results just mentioned by \cite{Halpern2000} and by \cite{Eiter}---see also \cite{darwiche2022causal}---could be said to lend further support to the claim that causal reasoning is no more difficult (in the sense of computational complexity) than purely probabilistic reasoning.

\subsubsection*{Overview of the paper} 
In the next two sections (\S\ref{section:into} and \S\ref{section:complexity}), we introduce the languages and the notions from computational complexity needed to state Theorem~\ref{characterization} more formally. The proof of this main result appears in \S\ref{section:mainresult}. Finally, in \S\ref{section:conclusion} we zoom out to consider what our results show about the relationship between probabilistic and causal reasoning, as well as consider a number of outstanding problems in this domain. In our presentation we assume no prior knowledge of causal modeling, complexity theory, or  probability logic. Only elementary logic and probability are presupposed. 

\section{Introducing Causal and Probabilistic Languages} \label{section:into}

In this section, we introduce the syntax and semantics for a series of  probabilistic and causal languages. With a precise syntax and semantics in hand, we illustrate that these languages form an expressive hierarchy.

\subsection{Syntax}\label{syntax}

Let $\mathbf{V}$ be a (possibly infinite) collection, representing the (endogenous) random variables under consideration. Informally, these are the variables that we may want to observe, change, query, or otherwise reason about explicitly. 

For each variable $V \in \mathbf{V}$, let $\text{Val}(V)$ denote the finite signature (range) of $V$. For example, for two binary variables we have $\mathbf{V} = \{X, Y\}$ with $\text{Val}(X) = \text{Val}(Y) = \{0, 1\}$. We introduce the following deterministic languages
\begin{align*}
    \mathcal{L}_{\text{int}}   &:= \top \;|\; V = v \;|\; \mathcal{L}_{\text{int}} \land \mathcal{L}_{\text{int}} & V \in \mathbf{V}, v \in \text{Val}(V)\\
    \mathcal{L}_{\mathrm{prop}}  &:= V = v \;|\; \neg \mathcal{L}_{\mathrm{prop}} \;|\; \mathcal{L}_{\mathrm{prop}} \land \mathcal{L}_{\mathrm{prop}} & V \in \mathbf{V}, v \in \text{Val}(V)\\
    \mathcal{L}_{\text{full}}  &:= [\mathcal{L}_{\text{int}}] \mathcal{L}_{\mathrm{prop}} \;|\; \neg \mathcal{L}_{\text{full}} \;|\; \mathcal{L}_{\text{full}} \land \mathcal{L}_{\text{full}}.
\end{align*}
Choose either $\mathcal{L}_{\mathrm{prop}}$ or $\mathcal{L}_{\mathrm{full}}$ as the \emph{base language} $\mathcal{L}$. The former is essentially a propositional language with extended ranges, while the latter is a causal conditional language. The semantics of these formulas will be introduced in \S\ref{section:semantics}, but intuitively we can interpret a formula of $\mathcal{L}_{\mathrm{full}}$, such as $[X=1]Y=0$, as expressing a subjunctive conditional: were $X$ to take on value $1$, then $Y$ would come to have value $0$. We understand the conditional causally, in a sense to be made precise below.

So-called \emph{terms} over the base language are the main ingredient of our probabilistic languages.
The most basic {term} is $\probsymbol(\delta)$ for $\delta \in \mathcal{L}$, representing the probability of $\delta$.
By varying the composite terms admitted, we can define polynomial, conditional, linear, and comparative languages. Where $\delta, \delta' \in \mathcal{L}$ are formulas of $\mathcal{L}$:
\begin{align*}
    T_{\mathrm{poly}}(\mathcal{L}) &\text{ is generated by the grammar }\mathbf{t} :=  \probsymbol(\delta) \; | \; \mathbf{t} + \mathbf{t}^\prime \; | \; \mathbf{t} \cdot \mathbf{t}^\prime \\
    T_{\mathrm{cond}}(\mathcal{L})&\text{ is generated by the grammar } \mathbf{t} :=   \probsymbol(\delta \; | \; \delta^\prime) \\
   T_{\mathrm{lin}}(\mathcal{L}) &\text{ is generated by the grammar } \mathbf{t} :=  \probsymbol(\delta) \; | \; \mathbf{t} + \mathbf{t}^\prime \\ T_{\mathrm{comp}}(\mathcal{L}) &\text{ is generated by the grammar }
    \mathbf{t} :=  \probsymbol(\delta)  
\end{align*}
We define for each $* \in \{\mathrm{comp}, \mathrm{lin},\mathrm{cond},\mathrm{poly}\}$ the causal and purely probabilistic languages:
\begin{align*}
    \mathcal{L}_{\mathrm{prob}}^* &:= \mathbf{t} \geq \mathbf{t} \; | \; \neg \mathcal{L}_{\mathrm{prob}}^* \; | \; \mathcal{L}_{\mathrm{prob}}^* \land \mathcal{L}_{\mathrm{prob}}^* &\mathbf{t} \in T_{*}(\mathcal{L}_{\mathrm{prop}}).\\
    \mathcal{L}_{\mathrm{causal}}^* &:= \mathbf{t} \geq \mathbf{t} \; | \; \neg \mathcal{L}_{\mathrm{causal}}^* \; | \; \mathcal{L}_{\mathrm{causal}}^* \land \mathcal{L}_{\mathrm{causal}}^* &\mathbf{t} \in T_{*}(\mathcal{L}_{\mathrm{full}}).
\end{align*}
Several of these probabilistic languages have appeared in the literature. For instance, $\mathcal{L}^{\mathrm{poly}}_{\mathrm{prob}}$ appeared already in early work by \cite{Scott1966AssigningPT}, while $\mathcal{L}^{\mathrm{lin}}_{\mathrm{prob}}$ was introduced  explicitly by \cite{fagin1990logic}. The language $\mathcal{L}^{\mathrm{poly}}_{\mathrm{causal}}$ was introduced and studied recently in \cite{ibeling2020probabilistic} (see also \citealt{bareinboim2020pearl} and \citealt{Eiter}). Many of these languages, however, have not yet received explicit treatment.

\subsection{Semantics} \label{section:semantics}

\subsubsection{Structural Causal Models}

The semantics for all of these languages will be defined relative to  \emph{structural causal models}, which can be understood as a very general framework for encoding \emph{data-generating processes}. In addition to the endogenous variables $\mathbf{V}$, structural causal models also employ \emph{exogenous variables} $\mathbf{U}$ as a source of random variation among endogenous settings. For extended introductions, see, e.g., \cite{Pearl2009,bareinboim2020pearl}.


\begin{definition}
A \textit{structural causal model} (SCM) $\mathfrak{M}$ is a tuple $\mathfrak{M} = (\mathcal{F}, \prob, \mathbf{U}, \mathbf{V})$, with: \begin{enumerate}[label=(\alph*)]
    \item $\mathbf{V}$ a set of \textit{endogenous variables}, with each $V \in \mathbf{V}$ taking on possible values $\text{Val}(V)$,
    \item $\mathbf{U}$ a set of \textit{exogenous variables}, with each $U \in \mathbf{U}$ taking on possible values $\text{Val}(U)$,
    \item $\mathcal{F} =\{f_V\}_{V \in \textbf{V}}$ a set of \emph{structural functions}, such that $f_V$ determines the value of $V$ given the values of the exogenous variables $\mathbf{U}$ and those of the other endogenous variables $V' \in \mathbf{V}$, and
    \item $\prob$ a probability measure on a $\sigma$-algebra $\sigma(\mathbf{U})$ on $\mathbf{U}$.
\end{enumerate}
Here we will assume for convenience that $\text{Val}(V)$ and $\text{Val}(U)$ are all finite. 
\end{definition}
In addition, we adopt the common assumption that our SCMs are recursive:
\begin{definition}
A SCM $\mathfrak{M}$ is \textit{recursive} if there is a well-order $\prec$ on $\mathbf{V}$ such that $\mathcal{F}$ respects $\prec$ in the following sense: for any $V \in \mathbf{V}$, whenever $\mathbf{v}_1,\mathbf{v}_2:V \mapsto \text{Val}(V)$ have the property that $\mathbf{v}_1(V') = \mathbf{v}_2(V')$ for all $V' \prec V$, we are guaranteed that $f_V(\mathbf{v}_1,\mathbf{u}) = f_V(\mathbf{v}_2,\mathbf{u})$.
\end{definition}

Intuitively, $\mathfrak{M}$ is recursive if for all $V \in \textbf{V}$, the function $f_V$ ensures that the value of $V$ is determined only by the exogenous random variables $U \in \mathbf{U}$ and endogenous random variables $V^\prime \in \textbf{V}$ for which $V^\prime \prec V$. Thus in a recursive model $\mathfrak{M}$, the probability measure $\prob$ on $\sigma(\textbf{U})$ induces a joint probability distribution $\prob(\textbf{V})$ over values of the variables $V \in \textbf{V}$.

Causal interventions represent the result of a \emph{manipulation} to the causal system, and  are defined in the standard way (e.g., \citealt{spirtes2000causation, Pearl2009}):
\begin{definition}
An \textit{intervention} is a partial function $i : V \mapsto \text{Val}(V)$. It specifies variables $\text{dom}(i) \subseteq \mathbf{V}$ to be held fixed and the values to which they are fixed. An intervention $i$ induces a mapping, also denoted $i$, of systems of equations $\mathcal{F} = \{f_V\}_{V \in \textbf{V}}$, such that $i (\mathcal{F})$ is identical to $\mathcal{F}$, but with $f_V$ replaced by the constant function $f_V(\cdot) = i(V)$ for each $V \in \text{dom}(i)$. Similarly, where $\mathfrak{M}$ is a model with equations $\mathcal{F}$, we write $i(\mathfrak{M})$ for the model which is identical to $\mathfrak{M}$ but with the equations $i(\mathcal{F})$ in place of $\mathcal{F}$.
\end{definition}

In order to guarantee that interventions lead to a well-defined semantics, we work with structural causal models which are measurable:

\begin{definition}
We say that $\mathfrak{M}$ is \textit{measurable} if under every finite intervention $i$, the joint distribution $\prob(\mathbf{V})$ associated with the model $i(\mathfrak{M})$ is well-defined.
\end{definition}

For measurable models, one can define a notion of causal influence:

\begin{definition}\label{defn:causalinfluence}
A model $\mathfrak{M}$ \textit{induces the influence relation} $V_i \rightsquigarrow V_j$ when there exist values $v, v^\prime \in \text{Val}(V_j)$ and interventions $\alpha,\alpha^\prime$ differing only in the value they impose upon $V_i$ for which\footnote{The truth definition for $\models$ is introduced formally below in \S\ref{sec:truth}.}
\[
\mathfrak{M} \models \probsymbol\big([\alpha] V_j = v \land [\alpha^\prime] V_j = v^\prime\big)  > 0.
\]
Given an enumeration of variables $V_1,...,V_n$ compatible with a well-order $\prec$, the model $\mathfrak{M}$ is \textit{compatible with} $\prec$ when it induces no instance $V_i \rightsquigarrow V_j$ with $i > j$.
\end{definition}

To illustrate the preceding definitions, we return to the front door graph shown in Example~\ref{ex:do}, and demonstrate an example of a SCM that is compatible with this graph:

\begin{example}\label{ex:scm}
Consider the SCM $\mathfrak{M} = (\mathcal{F}, \prob, \mathbf{U}, \mathbf{V})$, with the exogenous $\mathbf{U} = \{U, U_X, U_Y, U_Z\}$, each of which has probability $\nicefrac{1}{2}$ of being 1 and probability $\nicefrac{1}{2}$ of being 0, and with three endogenous variables $\mathbf{V} = \{X, Y, Z\}$. The equations $\mathcal{F} = \{f_V\}_{V \in \mathbf{V}}$ are given by
\begin{align*}
    f_X (U_X, U) &= U\land U_X\\
    f_Z(X, U_Z) & = X \land U_Z\\
    f_Y(Z, U_Y, U) &= Z \land U \land U_Y
\end{align*}
We observe that $\mathfrak{M}$ is measurable and recursive with the ordering $\prec$ given by $X \prec Z \prec Y$. Further, $X \rightsquigarrow Z$ and $Z \rightsquigarrow Y$, so that $\mathfrak{M}$ indeed realizes the front door graph and is compatible with $\prec$.
\end{example}

\subsubsection{Interpretations of Terms and Truth Definitions} \label{sec:truth}

It suffices to give the semantics for $\mathcal{L}_{\mathrm{causal}}^{\mathrm{poly}}$, since this language includes all of the other languages introduced above. A model is a recursive and measurable SCM $\mathfrak{M} =(\mathcal{F}, \prob, \mathbf{U},\mathbf{V})$. For each assignment $\mathbf{u}: U \mapsto \text{Val}(U)$ of values to exogenous variables, each $V \in \mathbf{V}$, and each $v \in \text{Val}(V)$, we define $\mathcal{F}, \textbf{u} \models V = v$ if the equations $\mathcal{F}$ together with the assignment $\textbf{u}$ assign the value $v$ to $V$. Conjunction and negation are defined in the usual way, giving semantics for $\mathcal{F}, \textbf{u} \models \beta$ for any $\beta \in \mathcal{L}_{\mathrm{prop}}$. If $\mathcal{F}, \mathbf{u} \models \beta$ holds for all $\mathbf{u}$, then we simply write $\mathcal{F} \models \beta$. When the relation $\mathcal{F},\mathbf{u} \models \beta$ does not depend on $\mathbf{u}$ at all---that is, we have $\mathcal{F}, \mathbf{u} \models \beta$ iff $\mathcal{F}, \mathbf{u'} \models \beta$ for all $\mathbf{u},\mathbf{u'}$ and all formulas $\beta$---we say that the equations $\mathcal{F}$ are deterministic. For $\beta,\beta^\prime \in \mathcal{L}_{\mathrm{prop}}$, we write $\beta \models \beta^\prime$ when $\mathcal{F} \models \beta \rightarrow \beta^\prime$ for all $\mathcal{F}$, where material implicaiton is defined in the usual way.

For each intervention $\alpha \in \mathcal{L}_{int}$ and each $\beta \in \mathcal{L}_{\mathrm{prop}}$, we define $\mathcal{F}, \mathbf{u} \models  [\alpha] \beta$ iff $i_\alpha(\mathcal{F}), \textbf{u}\models \beta$, where $i_\alpha$ is the intervention which effects the assignments described by $\alpha$. We also allow that $\alpha$ may be the trivial intervention $\top$, in which case we simply write $\beta$ instead of $[\alpha]\beta$. We define
\[
\big\llbracket \probsymbol(\epsilon)\big\rrbracket_\mathfrak{M} = \prob\big(\{\mathbf{u} : \mathcal{F},\textbf{u} \models \epsilon\}\big).
\]
For conditional probability terms we define $\big\llbracket \probsymbol(\delta | \delta^\prime)\big\rrbracket_\mathfrak{M} =1 $ when $\big\llbracket \probsymbol(\delta^\prime)\big\rrbracket_\mathfrak{M}=0$ and using the above definition and the usual ratio definition otherwise. For two terms $\textbf{t}_1,\textbf{t}_2$, we define $\mathfrak{M} \models \textbf{t}_1 \geq \textbf{t}_2$ iff $\llbracket \textbf{t}_1\rrbracket_\mathfrak{M} \geq \llbracket \textbf{t}_2\rrbracket_\mathfrak{M} $. The semantics for negation and conjunction are defined in the usual way, giving a semantics for $\mathfrak{M}\models\varphi$ for any $\varphi \in \mathcal{L}_{\mathrm{causal}}^{\mathrm{poly}}$.

With this semantics, probability behaves as expected. For example, we have the following validity for any $\epsilon_1,\epsilon_2$:
\[
\mathsf{Add.} \quad\probsymbol(\epsilon_1 \land \epsilon_2) + \probsymbol(\epsilon_1 \land \neg \epsilon_2) = \probsymbol(\epsilon_1).
\]

Causal interventions behave as expected as well. Indeed, fix any model $\mathfrak{M}$ with equations $\mathcal{F} $, any variable $V \in \textbf{V}$, and any assignment $\textbf{u}$ of values to the exogenous variables. Then $V$ takes on at least and at most one value upon the intervention $\alpha$: this is trivial if $\alpha$ intervenes on $V$, and it otherwise follows immediately from the fact that once $\textbf{u}$ is fixed, the values of all variables are determined by the equations $i_\alpha(\mathcal{F})$. In other words, in the language $\mathcal{L}_{\mathrm{causal}}^*$ for any $* \in \{\mathrm{comp, lin, cond, poly}\}$, we have the validity for all $\mathfrak{M}$ and $\textbf{u}$:
\[
\mathsf{Def.} \quad \bigwedge_{\substack{v,v^\prime \in \text{Val}(V)\\v \neq v^\prime}} \neg [\alpha] (V = v \land V = v^\prime) \land \bigvee_{v \in \text{Val}(V)} [\alpha] (V =v).
\] More generally, for each $\alpha \in \mathcal{L}_{\mathrm{int}}$, the indexed box $[\alpha]$ can be thought of as a \emph{normal}, \emph{functional} modal operator. 

Having introduced the syntax and semantics for several languages and pointed to some basic validities, we recall in the next subsection various results and examples that illustrate the expressive relationships between these languages.

\subsection{A Two-Dimensional Expressive Hierarchy}

\begin{definition}
For a formula $\varphi$ in any of the languages just introduced, let Mod$(\varphi) = \{\mathfrak{M}: \mathfrak{M} \models \varphi\}$ be the class of its models. For two languages $\mathcal{L}_{1}$ and $\mathcal{L}_{2}$, we say that $\mathcal{L}_{2}$ is \emph{at least as expressive as} $\mathcal{L}_{1}$ if for every $\varphi \in \mathcal{L}_1$ there is some $\psi \in \mathcal{L}_2$ such that Mod$(\varphi) = $ Mod$(\psi)$. We say $\mathcal{L}_{2}$ is \textit{strictly more expressive than }$\mathcal{L}_{1}$ if $\mathcal{L}_{2}$ is at least as expressive as $\mathcal{L}_{1}$ but not vice versa.
\end{definition}

In this section, mostly rehearsing familiar results and examples, we illustrate that the expressivity of the languages $\mathcal{L}^*$ for $\mathcal{L} \in \{\mathcal{L}_{\mathrm{prob}}, \mathcal{L}_{\mathrm{causal}}\}$ and $* \in \{\mathrm{comp, lin, cond, poly}\}$ form an expressive hierarchy along two axes. First, the purely probabilistic language $\mathcal{L}_{\mathrm{prob}}^*$ is always less expressive than the corresponding causal language $\mathcal{L}_{\mathrm{causal}}^*$. Second, $\mathcal{L}^{\mathrm{comp}}$ is less expressive than both $\mathcal{L}^{\mathrm{lin}}$ and $\mathcal{L}^{\mathrm{cond}}$, both of which are less expressive than the language $\mathcal{L}^{\mathrm{poly}}$. Where each arrow indicates a strict increase in expressivity, the hierarchy can be shown graphically:\footnote{The arrow in the center of these squares is meant to indicate that $\mathcal{L}_{\mathrm{prob}}^*$ is less expressive than $\mathcal{L}_{\mathrm{causal}}^*$ for any choice of $* \in \{\mathrm{comp, lin, cond, poly}\}$.}

\begin{center}
    \begin{tikzpicture}
    \node[vertex] (a1) at (0,0) {$\mathcal{L}^{\mathrm{comp}}_{\mathrm{prob}}$};
    \node[vertex] (a2) at (0,2) {$\mathcal{L}^{\mathrm{cond}}_{\mathrm{prob}}$};
    \node[vertex] (a3) at (2,0) {$\mathcal{L}^{\mathrm{lin}}_{\mathrm{prob}}$};
    \node[vertex] (a4) at (2,2) {$\mathcal{L}^{\mathrm{poly}}_{\mathrm{prob}}$};
    
    \node[vertex] (b1) at (4,0) {$\mathcal{L}^{\mathrm{comp}}_{\mathrm{causal}}$};
    \node[vertex] (b2) at (4,2) {$\mathcal{L}^{\mathrm{cond}}_{\mathrm{causal}}$};
    \node[vertex] (b3) at (6,0) {$\mathcal{L}^{\mathrm{lin}}_{\mathrm{causal}}$};
    \node[vertex] (b4) at (6,2) {$\mathcal{L}^{\mathrm{poly}}_{\mathrm{causal}}$};
    
    \draw[edge] (a1) -- (a2);
    \draw[edge] (a1) -- (a3);
    \draw[edge] (a2) -- (a4);
    \draw[edge] (a3) -- (a4);
    \draw[edge] (1,1) -- (5,1);
    \draw[edge] (b1) -- (b2);
    \draw[edge] (b1) -- (b3);
    \draw[edge] (b2) -- (b4);
    \draw[edge] (b3) -- (b4);
    \end{tikzpicture}
\end{center}

\subsubsection{First Axis: From Probabilistic to Causal}

To illustrate the expressivity of causal as opposed to purely probabilistic languages, we recall a variation by \cite{bareinboim2020pearl} on an example due to \cite{Pearl2009}:

\begin{example}[Causation without correlation]
Let $\mathfrak{M}_1 = (\mathcal{F}, \prob, \mathbf{U}, \mathbf{V})$, where $\textbf{U}$ contains two binary variables $U_1,U_2$ such that $\prob(U_1) = \prob(U_2) =\nicefrac{1}{2}$, and $\textbf{V}$ contains two variables $V_1, V_2$ such that $f_{V_1} = U_1$ and $f_{V_2} = U_2$. Then $V_1$ and $V_2$ are independent. Having observed this, one could not conclude that $V_1$ has no causal effect on $V_2$; indeed, consider the model $\mathfrak{M}^\prime$, which is like $\mathfrak{M}$, except with the mechanisms:
\begin{align*}
    f_{V_1} &= \mathbf{1}_{U_1 = U_2}\\
    f_{V_2} &= U_1 + \mathbf{1}_{V_1=1, U_1 = 0, U_2=1}.
\end{align*} Here $\mathbf{1}_S$ is the indicator function for statement $S$, equal to $1$ if $S$ holds and $0$ otherwise. 
In this case $\prob_{\mathfrak{M}}(V_1,V_2) = \prob_{\mathfrak{M}^\prime}(V_1,V_2)$, so that the models are indistinguishable in any of the probabilistic languages $\mathcal{L}_{\mathrm{prob}}^*$. However, the models are distinguishable in $\mathcal{L}_{\mathrm{causal}}^{\mathrm{comp}}$, and so in all of the other causal languages. Indeed, note that $\prob_{\mathfrak{M}}\big([V_1=1] V_2=1\big)= \nicefrac{1}{2}$ while $\prob_{\mathfrak{M}^\prime}\big([V_1=1] V_2=1\big)= \nicefrac{3}{4}$. Then, for instance, the following statement
\begin{align*}
    \prob\big([V_1 = 1] V_2 = 1\big) = \prob\big( [V_1 = 1] V_2 = 0\big) 
\end{align*}
belongs to $\mathcal{L}_{\mathrm{causal}}^{\mathrm{comp}}$ and distinguishes $\mathfrak{M}$ from $\mathfrak{M}^\prime$. \label{ex-cht}
\end{example}

As shown in \cite{bareinboim2020pearl} (cf. also \citealt{suppes:zan81}), the pattern in Example \ref{ex-cht} is universal: for any model $\mathfrak{M}$ it is \emph{always} possible to find some $\mathfrak{M}^\prime$ that agrees with $\mathfrak{M}$ on all of $\mathcal{L}_{\mathrm{prob}}^{\mathrm{poly}}$ but disagrees on $\mathcal{L}_{\mathrm{causal}}^{\mathrm{poly}}$.
\footnote{The Causal Hierarchy Theorem of \cite{bareinboim2020pearl} (refer to \citealt{II21} for a topological version, enabling the relevant generalization to infinite $\mathbf{V}$) involves an intermediate language between $\mathcal{L}_{\mathrm{prob}}^{\mathrm{poly}}$ and $\mathcal{L}_{\mathrm{causal}}^{\mathrm{poly}}$, capturing the type of causal information revealed by controlled experiments. Even this three-tiered hierarchy is strict, and in fact one can go further to obtain an infinite hierarchy of increasingly expressive causal languages between $\mathcal{L}_{\mathrm{prob}}^{\mathrm{poly}}$ and $\mathcal{L}_{\mathrm{causal}}^{\mathrm{poly}}$. Because we are showing that there is a complexity collapse even from the most expressive to the least expressive systems, we are not concerned in the present work with these intermediate languages.}

\begin{theorem} $\mathcal{L}_{\mathrm{causal}}^{\mathrm{poly}}$ is more expressive than $\mathcal{L}_{\mathrm{prob}}^{\mathrm{poly}}$. What is stronger, \emph{no} $\mathcal{L}_{\mathrm{prob}}^{\mathrm{poly}}$-theory (i.e., maximally consistent set in this language) uniquely determines a $\mathcal{L}_{\mathrm{causal}}^{\mathrm{poly}}$-theory. 
\end{theorem}

\subsubsection{Second Axis: From Qualitative to Quantitative}

Focusing just on probabilistic languages, we will show that $\mathcal{L}_{\mathrm{prob}}^{\mathrm{comp}}$ is less expressive than both $\mathcal{L}_{\mathrm{prob}}^{\mathrm{lin}}$ and $\mathcal{L}_{\mathrm{prob}}^{\mathrm{cond}}$, and that both of these are less expressive than the language $\mathcal{L}_{\mathrm{prob}}^{\mathrm{poly}}$. In each case, it suffices to give two measures $\prob_1(\textbf{V})$ and $\prob_2(\textbf{V})$ which are indistinguishable in the less expressive language but which can be distinguished by some statement in the more expressive one.

\paragraph{Comparative probability.} First, we claim that $\mathcal{L}_{\mathrm{prob}}^{\mathrm{comp}}$ is less expressive than $\mathcal{L}_{\mathrm{prob}}^{\mathrm{lin}}$. Suppose we have just a single binary variable $X$, abbreviating $X=1$ by $q$ and $X=0$ by $\neg q$. Then let $\prob_1(q) = \nicefrac{2}{3}$ so that $\prob_1(\neg q) = \nicefrac{1}{3}$, and let $\prob_2(q) = \nicefrac{3}{5}$ so that $\prob_2(\neg q)= \nicefrac{2}{5}$. The qualitative order on the four events $q, \neg q, \top, \bot$ is the same, but, for instance, $\prob_1(q) = \prob_1(\neg q)+\prob_1(\neg q)$, while $\prob_2(q) \neq \prob_2(\neg q)+\prob_2(\neg q)$.

Next, we recall an example due to \cite{luce1968numerical}, which shows that $\mathcal{L}_{\mathrm{prob}}^{\mathrm{comp}}$ is less expressive than $\mathcal{L}_{\mathrm{prob}}^{\mathrm{cond}}$. Let $p,q,r$ each be events corresponding to the three possible values taken by a random variable. Consider the measures $\prob_1(p) = \nicefrac{5}{9}, \prob_1(q) = \nicefrac{3}{9}, \prob_1(r) = \nicefrac{1}{9}$ and $\prob_2(p) = \nicefrac{6}{9}, \prob_2(q) = \nicefrac{2}{9}, \prob_2(r) = \nicefrac{1}{9}$. Then the two orders are the same, because for $i \in [2]$ $$\prob_i(\top) > \prob_i(p\vee q) > \prob_i(p \vee r) > \prob_i(p) > \prob_i(q \vee r) > \prob_i(q) > \prob_i(r) > \prob_i(\bot).$$ However, the conditional probabilities differ: $\prob_1(r | q \lor r) < \prob_1(q | p \lor q)$, while  $\prob_2(r | q \lor r) > \prob_2(q | p \lor q)$. In other words, the measures $\prob_1$ and $\prob_2$ are indistinguishable in $\mathcal{L}_{\mathrm{prob}}^{\mathrm{comp}}$ but distinguishable in $\mathcal{L}_{\mathrm{prob}}^{\mathrm{cond}}$. 

\paragraph{Polynomials in probabilities.} To show that $\mathcal{L}_{\text{prob}}^{\text{lin}}$ is less expressive than $\mathcal{L}_{\text{prob}}^{\text{poly}}$, we simply identify a formula $\varphi \in \mathcal{L}_{\text{prob}}^{\text{poly}}$ such that there is no $\psi \in \mathcal{L}_{\text{prob}}^{\text{lin}}$ with Mod$(\varphi) = $ Mod$(\psi)$. For this we can take the example $\prob(A \wedge B)=\prob(\neg A \vee \neg B) \land \prob(A|B) = \prob(B)$. (This is in fact expressible already in $\mathcal{L}_{\text{prob}}^{\text{cond}}$.) This enforces that $\prob(B)= 1/\sqrt{2}$, while \cite{IIMM22} show that every formula in $\mathcal{L}_{\text{prob}}^{\text{lin}}$ has models in which every probability is rational.

Finally, we give an example to show that $\mathcal{L}^{\mathrm{cond}}_{\mathrm{prob}}$ is less expressive than $\mathcal{L}^\mathrm{poly}_{\mathrm{prob}}$. As above, let $p,q,r$ be events corresponding to possible values taken by a random variable. Define $\prob_1(p)= \nicefrac{3}{20}, \prob_1(q)= \nicefrac{4}{20}, \prob_1(r)=\nicefrac{13}{20}$, while $\prob_2(p)=\nicefrac{3}{20} - .03,\prob_2(q)=\nicefrac{4}{20} - .01,\prob_2(r)=\nicefrac{13}{20} + .04$. One can verify by exhaustion that all comparisons of conditional probabilities agree between $\prob_1$ and $\prob_2$, thus they are indistinguishable in $\mathcal{L}_{\text{prob}}^{\mathrm{cond}}$. At the same time, there are statements in $\mathcal{L}_{\text{prob}}^{\mathrm{poly}}$ in which the models differ. For example, $\prob_1(r)\prob_1(q) <  \prob_1(p)$, whereas $\prob_2(r)\prob_2(q)  > \prob_2(p)$. This shows that $\mathcal{L}_{\mathrm{prob}}^{\mathrm{cond}}$ is less expressive than $\mathcal{L}_{\mathrm{prob}}^{\mathrm{poly}}$. Further, we observe that $\prob_1, \prob_2$ can be distinguished in $\mathcal{L}_{\mathrm{prob}}^{\mathrm{lin}}$: $\prob_i(q) \geq 0.2$ for $i = 1$ but not for $i=2$, and this statement is equivalent to the statement in $\mathcal{L}_{\mathrm{prob}}^{\mathrm{lin}}$ that
\[
\underbrace{\prob_i(q) + .. + \prob_i(q)}_{10 \text{ times}}  \geq \prob_i(\top) + \prob_i(\top).
\]
Together, this observation and the earlier remark that $\prob(A \wedge B)=\prob(\neg A \vee \neg B) \land \prob(A|B) = \prob(B)$ is expressible in $\mathcal{L}_{\text{prob}}^{\text{cond}}$ show that $\mathcal{L}^{\mathrm{lin}}$ and $\mathcal{L}^{\mathrm{cond}}$ are incomparable in expressivity.

Summarizing the results of this section: 
\begin{theorem} $\mathcal{L}^{\mathrm{lin}}$ and $\mathcal{L}^{\mathrm{cond}}$ are  incomparable in expressive power. Both are strictly more expressive than $\mathcal{L}^{\mathrm{comp}}$ and strictly less expressive than $\mathcal{L}^{\mathrm{poly}}$.
\end{theorem}

\section{Introducing Computational Complexity} \label{section:complexity}

In this section, we introduce the ideas from complexity theory needed to state our main results. We denote by $\SAT_{\mathrm{prob}}^*, \SAT_{\mathrm{causal}}^*$ the satisfiability problems for $\mathcal{L}_{\mathrm{prob}}^*, \mathcal{L}_{\mathrm{causal}}^*$, respectively, where $* \in \{\mathrm{comp, lin, cond, poly}\}$. There are two key definitions:

\begin{definition}
Say that a map $\varphi \mapsto \psi$ \textit{preserves and reflects satisfiability} when $\varphi$ is satisfiable if and only if $\psi$ is satisfiable. Such a map is called a \textit{many-one reduction} of $\varphi$ to $\psi$. Such a map is said to \textit{run in polynomial time} if it is computable by a Turing machine in a number of time steps that is a polynomial function of the length $|\varphi|$ of the input formula. When the Turing machine is non-deterministic, the map is said to be non-deterministic as well; in this case we say that the reduction is an \textit{$\mathsf{NP}$-reduction.}
\end{definition}

\begin{definition}
A \textit{decision problem} maps an input, represented as a binary string, to an output ``yes'' or ``no.'' For example, $\SAT_{\mathrm{prob}}^*$ maps a standard encoding of the formula $\varphi \in \mathcal{L}_{\mathrm{prob}}^*$ to ``yes'' if it is  satsifiable and to ``no'' otherwise. When each member of a collection $\mathcal{C}$ of decision problems can be reduced via some deterministic, polynomial-time map to a particular decision problem $ c\in \mathcal{C}$, one says that the problem $c$ is \textit{$\mathcal{C}$-complete.} The class $\mathcal{C}$ of decision problems is called a \textit{complexity class}.
\end{definition}

 When a problem $c$ is complete for some complexity class, this means that the complexity class $\mathcal{C}$ fully characterizes the difficulty of the problem: the problem $c$ is at least as ``hard'' as any of the problems in $\mathcal{C}$, and it is itself in $\mathcal{C}$. Thus any two problems which are complete for a complexity class are equally hard, since each can be reduced in deterministic polynomial time to the other. Complete problems facilitate results relating complexity classes: to show that a class $\mathcal{C}$ is contained in another $\mathcal{C}^\prime$, it suffices to give deterministic, polynomial-time, many-one reduction from a problem $c$ which is complete for $\mathcal{C}$ to any problem $c^\prime \in \mathcal{C}^\prime$.

\cite{fagin1990logic} showed that $\SAT_{\mathrm{prob}}^{\mathrm{lin}}$ is complete for the complexity class $\NP$. That $\SAT_{\mathrm{prob}}^{\mathrm{comp}}$ is also $\NP$-complete follows quickly from this result and the Cook-Levin theorem \citep{cook1971complexity}, which says that Boolean satisfiability is $\NP$-complete as well. For clarity, we include these known results in the statement of our main result, which gives completeness results for all of the other probabilistic and causal languages defined above:

\begin{reptheorem}{characterization}
We characterize two sets of tasks:
\begin{enumerate}
    \item 
    $\SAT_{\mathrm{prob}}^{\mathrm{comp}}, \SAT_{\mathrm{prob}}^{\mathrm{lin}}, \SAT_{\mathrm{causal}}^{\mathrm{comp}}, \SAT_{\mathrm{causal}}^{\mathrm{lin}}$ are $\NP$-complete.
    \item
    $\SAT_{\mathrm{prob}}^{\mathrm{cond}}, \SAT_{\mathrm{prob}}^{\mathrm{poly}}, \SAT_{\mathrm{causal}}^{\mathrm{cond}}, \SAT_{\mathrm{causal}}^{\mathrm{poly}}$ are $\ETR$-complete.
\end{enumerate}
\end{reptheorem}

Since problems that are complete for a class are all equally hard, our main results imply that causal and probabilistic reasoning in these languages do not differ in complexity. In the remainder of this section, we introduce the complexity classes $\NP$ and $\ETR$. We note that the inclusions $\NP \subseteq \ETR \subseteq \mathsf{PSPACE}$ are known \citep{canny1988some}, where $\mathsf{PSPACE}$ is the set of problems solvable using polynomial space; it is an open problem whether either inclusion is strict. Further, $\NP$ and $\PSPACE$ are closed under many-one $\NP$-reductions, and \cite{ten2013data} show that $\ETR$ is also closed under many-one $\NP$-reductions:

\begin{definition}
A complexity class $\mathcal{C}$ is \textit{closed under many-one $\NP$ reductions} if to show that a problem is in $\mathcal{C}$, it suffices to find a polynomial-time $\NP$-reduction of the problem to one that is known to be in $\mathcal{C}$.
\end{definition}

\subsection{The Class $\NP$}

The class $\NP$ contains any problem that can be solved by a non-deterministic Turing machine in a number of steps that grows polynomially in the input size. Equivalently, it contains any problem solvable by a polynomial-time \textit{deterministic} Turing machine, when the machine is provided with a polynomial-size \textit{certificate}, which we think of as providing the solution to the problem, or ``lucky guesses.'' In this case we think of the deterministic Turing machine as a \textit{verifier,} tasked with ensuring that the certificate communicates a valid solution to the problem.

 Hundreds of problems are known to be $\NP$-complete. Among them are Boolean satisfiability and the decision problems associated with several natural graph properties, for example possession of a clique of a given size or possession of a Hamiltonian path. See \cite{ruiz2011survey} for a survey of such problems and their relations.

\subsection{The Class $\ETR$}

The Existential Theory of the Reals (ETR) contains all true sentences of the form
\[
\text{there exist } x_1,...,x_n \in \mathbb{R} \text{ satisfying }\mathcal{S},
\]
where $\mathcal{S}$ is a system of equalities and inequalities of arbitrary polynomials in the variables $x_1,...,x_n$. For example, one can state in ETR the existence of the golden ratio, which is the only root of the polynomial $f(x) =x^2 - x -1$ greater than one, by ``there exists $x >1$ satisfying $f(x) = 0$.'' The decision problem of saying whether a given formula $\varphi\in$ ETR is complete (by definition) for the complexity class $\ETR$.

The class $\ETR$ is the real analogue of $\NP$, in two senses. Firstly, the satisfiability problem that is complete for $\ETR$ features real-valued variables, while the satisfiability problems that are complete for $\NP$ typically feature integer- or Boolean-valued variables. Secondly, and more strikingly, \cite{erickson2020smoothing} recently showed that while $\NP$ is the class of decision problems with answers that can be verified in polynomial time by machines with access to unlimited integer-valued memory, $\ETR$ is the class of decision problems with answers that can be verified in polynomial time by machines with access to unlimited \textit{real-valued} memory.

As with $\NP$, a myriad of problems are known to be $\ETR$-complete. We include some examples that illustrate the diversity of such problems:
\begin{itemize}
    \item
    In graph theory, there is the $\ETR$-complete problem of deciding whether a given graph can be realized by a straight line drawing \citep{schaefer2013realizability}. 
    \item
    In game theory, there is the $\ETR$-complete problem of deciding whether an (at least) three-player game has a Nash equilibrium with no probability exceeding a fixed threshold \citep{bilo2017existential}.
    \item 
    In geometry, there is the $\ETR$-complete ``art gallery'' problem of finding the smallest number of points from which all points of a given polygon are visible \citep{abrahamsen2018art}.
     \item
     In machine learning, there is the $\ETR$-complete problem of finding weights for a neural network trained on a given set of data such that the total error of the network falls below a given threshold \citep{abrahamsen2021training}.
\end{itemize}

For discussions of further $\ETR$-complete problems, see \cite{schaefer2009complexity} and \cite{cardinal2015computational}.


\section{Our results} \label{section:mainresult}

In this section, we prove our main result, Theorem~\ref{characterization}. To do this, we first establish that one can reduce satisfiability problems for causal languages to corresponding problems for purely probabilistic languages.

\subsection{Reduction}

\begin{definition}\label{def-reduction-languages}
Fix a set $OP$ of operations on $\mathbb{R}$, and for a given placeholder set $S$, let $OP(S)$ be the set of terms generated by application of operations in $OP$ to members of $S$. Define
\begin{align*}
    \mathcal{L}_{\mathrm{prob}} &= \mathbf{t}_1 \geq \mathbf{t}_2 \; | \; \neg \mathcal{L}_{\mathrm{prob}} \; | \; \mathcal{L}_{\mathrm{prob}}  \land \mathcal{L}_{\mathrm{prob}}  & \mathbf{t}_i \in OP\big(\{\probsymbol(\epsilon) : \epsilon \in \mathcal{L}_{\mathrm{prop}}\}\big)\\
    \mathcal{L}_{\mathrm{causal}} &= \mathbf{t}_1 \geq \mathbf{t}_2 \; | \; \neg \mathcal{L}_{\mathrm{causal}} \; | \; \mathcal{L}_{\mathrm{causal}}  \land \mathcal{L}_{\mathrm{causal}}  & \mathbf{t}_i \in OP\big(\{\probsymbol(\epsilon) : \epsilon \in \mathcal{L}_{\mathrm{full}}\}\big)
\end{align*}
The semantics for these languages are restricted to recursive SEMs.
\end{definition}

\begin{proposition}[Reduction]\label{NP-reduction}
There exists a many-one $\NP$ reduction from $\SAT_{\mathcal{L}_{\mathrm{causal}}}$ to $\SAT_{\mathcal{L}_{\mathrm{prob}}}$.
\end{proposition}

We first give a prose overview of the main ideas underlying the reduction. Fix $\varphi \in \mathcal{L}_{\mathrm{causal}}$. The key observation is that the reduction is straightforward when every $\epsilon$ with $\probsymbol(\epsilon)$ mentioned in $\varphi \in \mathcal{L}_{\mathrm{causal}}$ is a complete state description, where a complete state description says, for each possible intervention and each variable, what value that variable takes upon that intervention. Indeed, complete state descriptions have three nice properties:
\begin{enumerate}
    \item \label{p1}
    \textit{Polynomial-time comparison to ordering.} One can easily check whether a complete state description implies influence relations conflicting with a given order $\prec$ on the variables appearing in it. Indeed, one simply reads which variables influence which variables off of the intervention statements appearing in $\epsilon$.
    \item \label{p2}
    \textit{Existence of model matching probabilities.} If a collection of complete state descriptions does not conflict with an order $\prec$, then any probability distribution on the descriptions $\epsilon$ has a recursive model that induces it; briefly, one can simply take a distribution over deterministic models for the mutually unsatisfiable descriptions $\epsilon$.
    \item \label{p3}
    \textit{Small model property.} At most $|\varphi|$ complete state descriptions are mentioned in $\varphi$, and so at most that many receive positive probability in any model satisfying $\varphi$.
\end{enumerate}

These properties will allow a reduction to go through. Indeed, fix $\varphi \in \mathcal{L}_{\mathrm{causal}}$. Given that $\varphi$ is satisfiable, one can request as an $\NP$ certificate an ordering $\prec$ and (relying on \#\ref{p3}) the small set of complete state descriptions receiving positive probability. One then checks (relying on \#\ref{p1}) that these descriptions do not conflict with $\prec$. Since $\varphi$ is satisfiable only if there exists a measure satisfying its inequalities, one can safely translate those inequalities into the probabilistic language, giving a satisfiable probabilistic formula $\psi$. If the probabilistic formula $\psi$ is satisfiable via some measure, one can (relying on \#\ref{p2}) infer a corresponding recursive model for the causal formula $\varphi$. Thus the map $\varphi \mapsto \psi$ preserves and reflects satisfiability.

As it turns out, the same reduction goes through in the general case, when the $\epsilon$ for which $\probsymbol(\epsilon)$ is mentioned in $\varphi$ need not be complete state descriptions. Roughly, the strategy is to simply replace every $\epsilon$ such that $\probsymbol(\epsilon)$ is mentioned in $\varphi \in \mathcal{L}_{\mathrm{causal}}$ with an equivalent disjunction of complete state descriptions. The primary complication with this strategy is that there are too many possible interventions, variables, and values those variables could take on; \textit{truly} complete state descriptions are exponentially long, making the reduction computationally intractable. To address this issue, we work with a \textit{restricted class} of state descriptions, which feature only the interventions, variables, and values appearing in the input formula $\varphi$:

\begin{definition}\label{def:restricted-class-state-descriptions}
Fix a formula $\varphi \in \mathcal{L}_{\mathrm{prop}} \cup \mathcal{L}_{\mathrm{causal}}$. Let $\mathcal{I}$ contain all interventions appearing in $\varphi$ and let $\mathbf{V}_\varphi$ denote all variables appearing in $\varphi$. For each variable $V \in \mathbf{V}_\varphi$, let $\text{Assignments}_\varphi(V)$ contain $V =v$ whenever $V=v$ or $V\neq v$ appears in $\varphi$, and let it also contain one assignment $V = v^*$ not satisfying either of these conditions. Let $\Delta_\varphi$ contain all possible interventions paired with all possible assignments, where the possibilities are restricted to $\varphi$:
\begin{align*}
    \Delta_\varphi = \Big\{\bigwedge_{\alpha \in \mathcal{I}} \; \Big( [\alpha] \bigwedge_{V \in \mathbf{V}_\varphi} \beta_V^\alpha\Big) \;  :\; \beta_V^\alpha \in \text{Assignments}_\varphi(V) \text{ for }V \in \textbf{V}_{\varphi}\Big\}
\end{align*}
Call $\bigwedge_{V \in \mathbf{V}_\varphi} \beta_V^\alpha$ the \textit{results} of the intervention $\alpha$, and $\beta_V^\alpha$ the \textit{result for $V$} of the intervention $\alpha$. We write $\alpha \in \delta$ when $\delta \in \Delta_{\varphi}$ as shorthand for $\alpha \in \mathcal{I}$. We write $V \in \alpha$ when $\alpha$ contains some assignment $V= v$.
\end{definition}

The following three lemmas confirm that even working with this restricted class of state descriptions, (versions of) the three nice properties outlined above are retained.

\begin{definition}
Fix a formula $\varphi \in \mathcal{L}_{\mathrm{prop}} \cup \mathcal{L}_{\mathrm{causal}}$ and $ \Delta^\prime \subseteq \Delta_\varphi$. Fix a well-order $\prec$ on $\textbf{V}_{\varphi}$. Enumerate the variables $V_1,...,V_n$ in $\textbf{V}_\varphi$ in a way consistent with $\prec$. The formula $\delta \in \Delta^\prime$ is \textit{compatible with} $\prec$ when there exists a model $\mathfrak{M}$ that assigns positive probability to $\delta$ and that is compatible with $\prec$. Define $\Delta_\prec$ to contain all $\delta \in \Delta_\varphi$ compatible with $\prec$.
\end{definition}

\begin{lemma}[Polytime Comparison to Ordering]\label{check-ordering-of-delta}
Fix $\varphi \in \mathcal{L}_{\mathrm{prop}} \cup \mathcal{L}_{\mathrm{causal}}$. Given a set $\Delta^\prime \leq |\varphi|$, one can check that $\Delta^\prime \subseteq \Delta_{\varphi}$ and that each $\delta \in \Delta^\prime$ is compatible with $\prec$ in time polynomial in $|\varphi|$.
\end{lemma}

This lemma shows that given some statement $\varphi$ and a set of formulas $\Delta^\prime$, one can efficiently (i.e. in polynomial time) check that the formulas $\delta \in \Delta^\prime$ satisfy two conditions. The first condition is that the formulas $\delta$ describe, in the fullest terms possible, the ways that $\varphi$ could be true (i.e. $\Delta^\prime \subseteq \Delta_\varphi$). The second is that the formulas $\delta $ do not rule out the causal influence relations specified by the order $\prec$, for example the relations $X \prec Z \prec Y$ induced by the model of smoking's effect on lung cancer discussed in Example~\ref{ex:do} and Example~\ref{ex:scm}.

\begin{proof}
Checking that $\Delta^\prime \subseteq \Delta_{\varphi}$ is fast, since one can simply scan $\varphi$ to make sure that $\varphi$ mentions precisely interventions mentioned in all $\delta \in \Delta^\prime$; that $\varphi$ mentions precisely the variables $V$ appearing in the results of every intervention in $\delta \in \Delta^\prime$; and that for each such variable $V$, at most of one of its assignments $V=v$ in $\Delta^\prime$ does not appear as an assignment or a negated assignment in $\varphi$.

We now give an algorithm to check whether $\delta \in \Delta^\prime$ is compatible with $\prec$. We first give prose and formal descriptions of the algorithm and then consider its runtime and correctness.

Order the variables $V_1,...,V_n$ in $\textbf{V}_\varphi$ in a way consistent with the well-order $\prec$. For each variable $V_i$ with $i \in [n]$, do the following. First, for each intervention $\alpha$ in $\delta$ that mentions $V_i$, confirm that the intervention leads to satisfiable results: if $\delta$ says that upon the intervention $\alpha$ which sets $V_i= v$, the variable $V_i$ takes a value $v^\prime \neq v$, we reject $\delta$, which necessarily has probability 0. Next, for each pair of interventions $\alpha, \alpha^\prime$ in $\delta$ which do not intervene on the value assigned to $V_i$, check whether both interventions result in the same assignments to variables $V_j$ for all $j < i$; we say that such interventions $\alpha,\alpha^\prime$ have \textit{agreement} on all $V_j$ for $j < i$. If this is the case, and yet $\delta$ says that these two interventions result in different values for $V_i$, reject $\delta$; since $V_i$ can depend only on the values of $V_j$ for $j < i$, when these values are constant, $V_i$ must be constant as well. Here is a formal description of the algorithm. We will write $V \in \alpha$ to denote that the variable $V$ appears (or is mentioned) in the intervention $\alpha$, i.e., that $V = v$ is a conjunct in $\alpha$ for some value $v$. \vspace{.1in}

\begin{algorithm}[H]
\SetAlgoLined
Order the variables $V_1,...,V_n$ in $\textbf{V}_\varphi$ according to $\prec$\\
\For{i in 1,...,n}{
    \For{intervention $\alpha$ in $\delta$ with $V_i =v$ appearing in $\alpha$}{
            \If{$V_i = v^\prime$ with $v \neq v^\prime$ appears in the conjunction of assignments following $\alpha$}{
                \Return{$\delta$ is unsatisfiable, and so incompatible with $\prec$}
            }
        
    }
    \For{interventions $\alpha, \alpha^\prime$ in $\delta$ agreeing on all $V_j$ for $j < i$, and such that $V_i \notin \alpha$ and $V_i \notin \alpha^\prime$}{
        \If{$\alpha$ results in $V_i = v$ and $\alpha^\prime$ results in $V_i = v^\prime$ with $v \neq v^\prime$}{
            \Return{$\delta$ is incompatible with $\prec$}
        
        }
        
    }
}
\Return{$\delta$ is compatible with $\prec$}

\caption{Check that $\delta \in \Delta^\prime$ is compatible with $\prec$}
\end{algorithm} \vspace{.1in}

Below, we show that the above algorithm indeed runs in time $\mathrm{poly}(|\varphi|)$ and is correct, but for clarity, let us step through its execution on some examples. Consider the input $\delta := [V_1= 0] V_1 = 1$. Then, by the first ``if'' clause in the algorithm, $\delta$ is rejected as unsatisfiable, since the intervention $[V_1= 0]$ leads to impossible results. For another example, let $\delta^\prime$ be the formula \begin{multline*}[V_1 = 1\land V_4 = 1] (V_1 = 1\land V_2 =0\land V_3=1\land V_4=1) \\ \land [V_1 =1\land V_4=  0] (V_1 = 1\land V_2 = 0\land V_3 = 0\land V_4= 0).\end{multline*} Then in the second ``if'' clause on the third iteration, $\delta^\prime$ is rejected as incompatible with $\prec$, because the interventions $\alpha = [V_1 = 1\land V_4 = 1]$ and $\alpha^\prime = [V_1 =1\land V_4=  0]$ do not intervene on $V_3$, result in the same values for $V_1$ and $V_2$, and do result in the same value for $V_3$, contradicting the fact that $V_3$'s value must depend only on those assigned to $V_1$ and $V_2$.

It is helpful in considering these examples and the runtime of the algorithm to consider the following table of values:

\begin{center}
\begin{tabular}{ |p{4cm}||p{1.5cm}|p{1.5cm}|p{1.5cm}|  p{1.5cm}|}
 \hline
 \multicolumn{5}{|c|}{Results of all interventions in the input formula $\delta^\prime$} \\
 \hline
 Intervention & $V_1$ & $V_2$ & $V_3$ & $V_4$\\
 \hline
 $\alpha = [V_1 =1\land V_4 = 1]$ & $V_1 = 1$ & $V_2 = 0$ & $V_3 = 1$ & $V_4 = 1$\\
 $\alpha^\prime = [V_1 =1\land V_4 = 0]$ & $V_1 = 1$ & $V_2 = 0$ & $V_3 = 0$ & $V_4 = 0$\\
 \hline
\end{tabular}
\end{center}

In effect, the second ``for'' loop over all interventions $\alpha ,\alpha^\prime$ constructs the above table, starting with the leftmost column $V_1$ and proceeding to the right. The algorithm rejects $\delta^\prime$ when two cells in the column $V_i$ and rows $\alpha$ and $\alpha^\prime$ (with $V_i\notin\alpha, \alpha^\prime$) do not assign the same value to $V_i$ but agree on all columns $V_j$ to the left. The restriction that $V_i$ does not appear in $\alpha$ or $\alpha^\prime$ must be included because distinct interventions $\alpha, \alpha^\prime$ can disagree on the values they impose on $V_i$ when intervening on it, regardless of the values assigned to $V_j$ with $j < i$; such disagreement does not constitute a violation of the ordering $\prec$.

Let us first confirm that this algorithm runs in time $\mathrm{poly}(|\varphi|)$ and then show its correctness. We observe that $\max\{|\delta|,n\} = \mathrm{poly}(|\varphi|)$. The algorithm contains an $O(n)$ loop over $V_1,...,V_n$ and two $O(|\delta|^2) $ loops over interventions. The work performed inside of these loops takes time $O(n \cdot |\delta|)$, since we are simply reading $\delta$ and checking values for the variables $V_j$ for all $j < i$, which can be stored in a lookup table (like the one above) of size $O(n \cdot |\delta|)$. Thus the runtime of the algorithm is indeed $\mathrm{poly}(|\varphi|)$.

Finally, we confirm that the algorithm is correct. Fix any $\delta \in \Delta^\prime$ and recall that $\delta$ is of the form
\[
\bigwedge_{\alpha} \; \Big( [\alpha] \bigwedge_{V \in \mathbf{V}_\varphi} \beta_V^\alpha\Big),
\]
where $\beta_V^\alpha \in \text{Assignments}_\varphi(V)$. First, suppose that the above algorithm declares $\delta$ compatible with $\prec$. We will inductively construct a deterministic model of equations $\mathcal{F} = \{f_{V_i} \}_{i \in [n]}$ and show that $\mathcal{F} \models \delta$ and $\mathcal{F}$ is compatible with $\prec$. Define $f_{V_1}$ to be the constant function sending all arguments to $\beta_{V_1}$, where $\beta_{V_1}$ is the value of $V_1$ upon any intervention $\alpha \in \delta$ with $V_1 \not\in \alpha$; the second ``for'' loop in the algorithm ensures that there is at most one such value, and if there is no such value, $\beta_{V_1}$ can be chosen arbitrarily. Then $f_{V_1} \models \bigwedge_{\alpha } [\alpha] \beta_{V_1}^\alpha$. Indeed, this holds by construction for $\alpha$ with $V_1 \not \in \alpha$, and it holds trivially for $\alpha$ with $V_1 \in \alpha$, because, by the first ``for'' loop, each $\alpha$ is compatible with its results. For the inductive step, define $f_{V_i}(V_1 = \beta_{V_1},...,V_{i-1}=\beta_{V_{i-1}}) = \beta_{V_i}$, where $\beta_{V_i}$ is the value of $V_i$ upon any intervention $\alpha \in \delta$ for which $V_i \not\in \alpha$ and $\beta_{V_j}^\alpha = \beta_{V_j}$ for all $j < i$; by the same reasoning, there is at most one such value, and if there is no such value, $\beta_{V_i}$ can be chosen arbitrarily. Then by the same reasoning, $f_{V_i} \models \bigwedge_{\alpha} [\alpha] \beta_{V_i}^\alpha$. Because this holds for all $i  \in [n]$, we have $\mathcal{F} \models \delta$. By construction, $\mathcal{F}$ is compatible with $\prec$, as desired.

Now, suppose that $\delta$ is compatible with $\prec$, so that $\delta$ is not self-contradictory and there exists some $\mathcal{F}=\{f_{V_i} \}_{i \in [n]}$ compatible with $\prec$ for which $\mathcal{F}\models \delta$. We claim that the above algorithm returns that $\delta$ is indeed compatible with $\prec.$ Suppose for a contradiction that on iteration $i$, the algorithm rejects $\delta$ as incompatible with $\prec$. Since $\delta$ is not self-contradictory, it follows by the definition of the algorithm that for some interventions $\alpha, \alpha^\prime$ (with $V_i \not \in \alpha, \alpha^\prime$) which agree on all $V_j$ for $j < i$, we have $[\alpha] V_i = v$ and $[\alpha^\prime] V_i = v^\prime$ with $v \neq v^\prime$. Let $\beta_j$ be the value such that the assignments $V_j = \beta_j$ for $j < i$ result from the interventions $\alpha$ and $\alpha^\prime$. Then
\[
\mathcal{F} \models [V_1 = \beta_1,...,V_{i-1} = \beta_{i-1}] ( V_i = v \land V_i \neq v),
\]
which is impossible.
\end{proof}

\begin{lemma}[Existence of Model Matching Probabilities]\label{model-from-prob}
Fix $\varphi \in \mathcal{L}_{\mathrm{prop}} \cup \mathcal{L}_{\mathrm{causal}}$, and suppose $\prob$ is a measure on $\Delta_{\prec} \subseteq \Delta_{\varphi}$ for some $\prec$. Then there is a model $\mathfrak{M}$ inducing the measure $\prob$ on $\Delta_{\prec}$, i.e., $\big\llbracket\probsymbol(\delta)\big\rrbracket_\mathfrak{M}= \prob(\delta)$ for all $\delta \in \Delta_\prec.$
\end{lemma}
\begin{proof}
Let us first define the model $\mathfrak{M}$ and then show that it is recursive. Let $\mathbf{V}_{\varphi}$ denote all variables appearing in $\varphi$. We define $\mathfrak{M} = \big(\mathcal{F}, \prob_{\mathfrak{M}}, \{U\}, \mathbf{V}_{\varphi}\big)$, where $\mbox{Val}(U) = \Delta_\varphi$ and $\prob_{\mathfrak{M}} (U = \delta) = \prob(\delta)$ for all $\delta \in \Delta_\varphi$. Enumerate the variables $\mathbf{V}_\varphi = \{V_1,...,V_n\}$ in a way consistent with $\prec$. Fix any $\delta \in \Delta_\varphi$. Recall that $\delta$ is of the form
\[
\bigwedge_{\alpha} \; \Big( [\alpha] \bigwedge_{i \in [n]} \beta_{V_i}^\alpha\Big),
\]
where $\beta_{V_i}^\alpha \in \text{Assignments}_\varphi(V_i)$. If $\delta$ is satisfiable, it has a model, i.e. a deterministic system of equations $\mathcal{F}^\delta = \{f_{V_i}^\delta\}_{i \in [n]}$ such that $\mathcal{F}^\delta \models \delta$. Turning now to define the equations $\mathcal{F} = \{f_{V_i}\}_{i \in [n]}$, for any assignment $\textbf{v}$ to the variables $V_j$ for $j < i$, put
\begin{align*}
    f_{V_i}(\mathbf{v}, U = \delta) = f_{V_i}^\delta(\mathbf{v}).
\end{align*}
By the above equations and mutual unsatisfiability of $\delta \in \Delta_{\prec}$, it follows that for all such $\delta$ $$\big\llbracket \probsymbol(\delta) \big\rrbracket_{\mathfrak{M}} = \prob_{\mathfrak{M}}[U = \delta] = \prob(\delta),$$ as required.

It remains for us to confirm that $\mathfrak{M}$ is recursive. We claim that the influence relationships $V_i \rightsquigarrow V_j$ induced by the model $\mathfrak{M}$ are simply those induced by the deterministic models\footnote{I.e., according to Definition \ref{defn:causalinfluence} when each deterministic model is thought of as a probabilistic model in which its respective system of functions is selected with no uncertainty.} $\mathcal{F}^\delta$. This would complete the proof, since by assumption we have $\delta \in \Delta_{\prec}$, so that $\mathcal{F}^\delta$ is compatible with $\prec$, and therefore $i < j$. Suppose that $\mathfrak{M}$ induces the influence relation $V_i \rightsquigarrow V_j$. Then for some interventions $\alpha, \alpha^\prime$ which disagree only on the value assigned to $V_i$, some assignment $\textbf{u}$ to $U$, and some distinct values $v, v^\prime$ of $V_j$, we have
\[
\mathcal{F} , \textbf{u} \models [\alpha] V_j = v \land [\alpha^\prime] V_j = v^\prime.
\]
Let $\delta$ be the value that $\textbf{u}$ assigns to $U$. We claim that
\[
\mathcal{F}^\delta \models [\alpha] V_j = v \land [\alpha^\prime] V_j = v^\prime.
\]
Indeed, this follows from the fact that $f_{V_i}(\textbf{v}, U= \delta) = f_{V_i}^\delta(\textbf{v})$ for all $i \in [n]$.
\end{proof}

\begin{lemma}[Small Model Property]\label{small-model}
Fix $\varphi \in \mathcal{L}_{\mathrm{prop}} \cup \mathcal{L}_{\mathrm{causal}}$. If $\varphi$ is satisfiable, then $\varphi$ has a small model, in the sense that the model assigns positive probability to at most $|\varphi|$ elements $\delta \in \Delta_{\varphi}$.
\end{lemma}
\begin{proof}
Since $\varphi$ is satisfiable, it has a recursive model $\mathfrak{M}$ with some order $\prec$. Given the existence of $\mathfrak{M}$, we claim there exists a small model $\mathfrak{M}_{\mathrm{small}}$. Indeed, consider the system of equations in the unknowns $\{\prob(\delta) : \delta \in \Delta_\prec\}$ given by
\begin{align*}
    \sum_{\delta \in \Delta_\prec} \prob(\delta) &= 1\\
    \sum_{\substack{\delta \in \Delta_\prec\\ \delta \models \epsilon}} \prob(\delta) &= \prob_\mathfrak{M}(\epsilon),  \text{ for each }\epsilon \text{ such that }\probsymbol(\epsilon) \text{ is mentioned in }\varphi.
\end{align*}
There are at most $|\varphi|$ equations, and since for each $\epsilon$, there exists some $\delta \in \Delta_\prec$ for which $\delta \models \epsilon$, the equations are non-trivial. Suppose for the moment that $\prob = \prob_\mathfrak{M}$ is a solution. Then by a fact of linear algebra (see Lemma~4.8 of \citealt{fagin1990logic}), since the at most $|\varphi|$ linear equations have a solution, they have a solution $\prob=\prob_{\mathrm{small}}$ in which at most $|\varphi|$ of the variables $\prob_{\mathrm{small}}(\delta)$ are nonzero. By Lemma~\ref{model-from-prob}, we then infer from the existence of $\prob_{\mathrm{small}}$ that the desired model $\mathfrak{M}_{\mathrm{small}}$ exists.

It remains to confirm that $\prob = \prob_\mathfrak{M}$ is indeed a solution to the above system of equations. To show this, we must show that for $\epsilon$ with $\probsymbol(\epsilon)$ mentioned in $\varphi$,
\[
\prob(\epsilon) = \sum_{\substack{\delta \in \Delta_\prec\\ \delta\models \epsilon}} \prob(\delta).
\]
By our choice of $\prec$, we know that the recursive model $\mathfrak{M}$ will assign probability 0 to all $\delta \not \in \Delta_\prec$. It thus suffices to show that the above holds when $\Delta_\prec$ is replaced with the larger set $\Delta_\varphi$. To do this, we will put $\epsilon$ into a more manageable form; afterwards, establishing the above equality will be relatively straightforward.

If $\epsilon$ mentions only one intervention $\alpha$, we claim that $\models \epsilon \leftrightarrow [\alpha] \beta$, where $\beta$ is an assignment of variables in $\varphi$ to various values. Indeed, negation and conjunction distribute over $[\alpha]$, in the sense that $\models \neg [\alpha] \beta \leftrightarrow  [\alpha] \neg\beta$ and $\models  [\alpha] \beta \land  [\alpha] \beta^\prime \leftrightarrow  [\alpha] (\beta \land \beta^\prime)$, so $[\alpha]$ can be assumed to appear on the outside. Further, since by the validity $\mathsf{Def}$, each variable takes one and only one value upon the intervention $\alpha$, we can replace $\beta$ with a disjunction over all assignments to all variables in $\varphi$ which agree with $\beta$. Let us use $\mathbf{v}_\varphi$ to denote such an assignment $\bigwedge_{V \in \mathbf{V}_\varphi} \beta_V$ where $\beta_V  \in \mathrm{Assignments}(V)$, as defined in Definition~\ref{def:restricted-class-state-descriptions}. Summing up:
\[
\models \epsilon \leftrightarrow [\alpha] \bigvee_{\mathbf{v}_{\varphi} \models \beta} \mathbf{v}_{\varphi}.
\]
The exact same ideas apply when $\epsilon$ mentions several interventions $[\alpha_i] \beta_i$ for $i \in [n]$, in which case
\[
\models \epsilon \leftrightarrow \bigwedge_i [\alpha_i] \bigvee_{\mathbf{v}_{\varphi} \models \beta_i} \mathbf{v}_{\varphi} \leftrightarrow \bigvee_{\mathbf{v}^i_\varphi \models \beta_i } \bigwedge_i [\alpha_i] \mathbf{v}_\varphi^i.
\]
Thus, since all interventions, variables, and assignments appearing in $\epsilon$ are mentioned by the $\delta \in \Delta_\varphi$, and one can always add trivial interventions $[\alpha] \top$, we see that it is a validity that $\epsilon$ is equivalent to a disjunction of formulas $\delta \in \Delta_\varphi$. Finally, we conclude with the observation that since the $\delta \in \Delta_\varphi$ are mutually unsatisfiable, additivity for the measure $\prob$ (according to which $\prob(\delta \lor \delta^\prime) = \prob(\delta) + \prob(\delta^\prime)$ for mutually unsatisfiable $\delta, \delta^\prime$) tells us that
\[
\prob(\epsilon) = \sum_{\substack{\delta \in \Delta_\varphi\\ \delta\models \epsilon}} \prob(\delta),
\]
as desired.
\end{proof}

With the lemmas in hand, we now give the desired reduction:

\begin{proof}[Proof of Proposition~\ref{NP-reduction}]
Fix a $\SAT_{\mathcal{L}_{\mathrm{causal}}}$ instance $\varphi$. We first describe the $\NP$ certificate and many-one reduction and then prove soundness and completeness. The $\NP$ certificate consists of an order $\prec$ on $\mathbf{V}_\varphi$ and a set of $\Delta^\prime$ of size at most $|\varphi|$. The reduction proceeds as follows. 
\begin{enumerate}
    \item 
    Check that $\Delta^\prime \subseteq \Delta_{\varphi}$ and that each $\delta \in \Delta^\prime$ is compatible with $\prec$.
    
    We note that by Lemma~\ref{check-ordering-of-delta}, this can be done in time polynomial in $|\varphi|$.
    \item
    Replace $\probsymbol(\epsilon)$ appearing in $\varphi$ with $\probsymbol\Big(\bigvee_{\substack{\delta \in \Delta^\prime\\\delta \models \epsilon}} f(\delta)\Big)$, where $f$ is a bijection between $\Delta^\prime$ and an arbitrary set of mutually unsatisfiable statements in $\mathcal{L}_{\mathrm{prop}}$. Call the resulting $\mathcal{L}_{\mathrm{prob}}$ formula $\varphi(\Delta^\prime)$.
    
    We note that checking $\delta \models \epsilon$ can be done in polynomial time, since $\delta$ is a complete description of the results of all interventions.
\end{enumerate}

\textbf{Completeness:} If $\varphi$ is satisfiable, by Lemma~\ref{small-model} it has a small model that assigns positive probability only to some $\Delta^\prime \subseteq \Delta_\prec$ for some ordering $\prec$, and the probabilities given by this model also solve $\varphi(\Delta^\prime)$. So the certificate exists and the reduction succeeds in producing a satisfiable formula.

\textbf{Soundness:} If $\varphi(\Delta^\prime)$ is satisfiable, it is solved by some measure $\prob$. This is a measure defined on $f(\Delta^\prime)$, and so on $\Delta^\prime \subseteq \Delta_\prec$. Thus since each $\delta \in \Delta^\prime$ is compatible with $\prec$, by Lemma~\ref{model-from-prob} there exists a model $\mathfrak{M}$ such that $\big\llbracket\probsymbol(\delta)\big\rrbracket_\mathfrak{M} = \prob(\delta)$ for $\delta\in\Delta_\varphi$. This $\mathfrak{M}$ is a model of the inequalities stated by $\varphi$ and is recursive, so $\varphi$ is satisfiable as well.
\end{proof}

\subsection{Characterization}

Now, we show our main result:

\begin{reptheorem}{characterization}
We characterize two sets of tasks:
\begin{enumerate}
    \item 
    $\SAT_{\mathrm{prob}}^{\mathrm{comp}}, \SAT_{\mathrm{prob}}^{\mathrm{lin}}, \SAT_{\mathrm{causal}}^{\mathrm{comp}}, \SAT_{\mathrm{causal}}^{\mathrm{lin}}$ are $\NP$-complete.
    \item
    $\SAT_{\mathrm{prob}}^{\mathrm{cond}}, \SAT_{\mathrm{prob}}^{\mathrm{poly}}, \SAT_{\mathrm{causal}}^{\mathrm{cond}}, \SAT_{\mathrm{causal}}^{\mathrm{poly}}$ are $\ETR$-complete.
\end{enumerate}
We can express these results in a diagram, which holds for $* \in \{\mathrm{prob, causal}\}$:

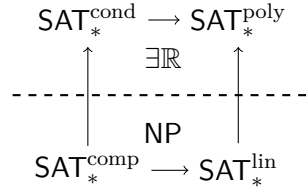
\begin{center}
    \begin{tikzpicture}
    \node[vertex] (1) at (0,0) {$\SAT_*^{\mathrm{comp}}$};
    \node[vertex] (2) at (0,2) {$\SAT_*^{\mathrm{cond}}$};
    \node[vertex] (3) at (2,0) {$\SAT_*^{\mathrm{lin}}$};
    \node[vertex] (4) at (2,2) {$\SAT_*^{\mathrm{poly}}$};
    \node[vertex] at (1,1.5) {$\ETR$};
    \node[vertex] at (1,.5) {$\NP$};
    \draw [thick,dashed] (-1,1) -- (3,1);
    
    \draw[edge] (1) -- (2);
    \draw[edge] (1) -- (3);
    \draw[edge] (2) -- (4);
    \draw[edge] (3) -- (4);
    \end{tikzpicture}
\end{center}
The line separates $\ETR$-complete problems from $\NP$-complete problems, and an arrow from one satisfiability problem to another indicates that any instance of the former problem is an instance of the latter.
\end{reptheorem}

We note that these results imply that there exists a many-one, polynomial-time, \textit{deterministic} many-one reduction from $\SAT_{\mathrm{causal}}^{*}$ to $\SAT_{\mathrm{prob}}^{*}$, for any $* \in \{\mathrm{comp, lin, cond, poly}\}$, whereas Proposition~\ref{NP-reduction} only gives a non-deterministic reduction. To illustrate, recall the model of smoking's effect on lung cancer discussed in Example~\ref{ex:do} and Example~\ref{ex:scm}. Consider again the task of determining whether smoking makes one more likely to possess lung cancer, given one's causal assumptions $\Gamma$ and one's observation of statistical correlation between smoking, tar deposits in the lungs, and lung cancer. In other words, the task is determine whether
\begin{equation} \Gamma + \mbox{Correlational data } \models \probsymbol\big([X = 1] Y =1\big) > \probsymbol\big([X = 0] Y =1 \big)\label{equation:causal},\end{equation} where the correlational data includes statements such as $\probsymbol(Y = 1 |X = 1) > \probsymbol(Y =1 | X  = 0)$ and $ 0.7 > \probsymbol(Y = 1| X=1) > 0.6$. The above result implies that this task is no more difficult than that of determining whether an analogous entailment
\begin{equation}
    \Gamma^\prime  + \mbox{Correlational data} \models \mbox{ Probabilistic conclusion }\label{equation:probabilistic}
\end{equation}
holds, given \textit{purely probabilistic} assumptions $\Gamma^\prime$. Indeed, given Equation~(\ref{equation:causal}), one can efficiently (i.e. in polynomial time) construct a probabilistic equation with the form of Equation~(\ref{equation:probabilistic}) such that both entailments have the same truth-value; the causal inference goes through if and only if the purely probabilistic inference goes through. 

To show the results in the second part of Theorem~\ref{characterization}, we borrow the following lemma from \cite{abrahamsen2018art}:

\begin{lemma}\label{etr-inv}
Fix variables $x_1,...,x_n$, and set of equations of the form $x_i + x_j = x_k$ or $x_i x_j = 1$, for $i,j,k \in [n]$. Let $\exists\mathbb{R}$-inverse be the problem of deciding whether there exist reals $x_1,...,x_n$ satisfying the equations, subject to the restrictions $x_i \in [\nicefrac{1}{2},2]$. This problem is $\ETR$-complete.
\end{lemma}

Here, for reasons of space we only outline the two steps in the proof of Lemma~\ref{etr-inv}. First, one shows that finding a real root of a degree 4 polynomial with rational coefficients is $\ETR$-complete, and then one repeatedly performs variable substitutions to get the constraints $x_i + x_j = x_k$ and $x_i x_j = 1$. Second, one shows that any such polynomial has a root within a closed ball about the origin, and then one shifts and scales this ball to contain exactly the range $[\nicefrac{1}{2},2]$.

With the lemma in hand, we show the theorem:

\begin{proof}[Proof of Theorem \ref{characterization}]
We begin with the first statement. Using the fact that $\SAT_{\mathrm{prob}}^{\mathrm{comp}}$ is $\NP$-hard, it suffices to show that $\SAT_{\mathrm{causal}}^{\mathrm{lin}}$ is inside $\NP$; indeed, since all of the satisfiability problems mentioned in the first statement include $\SAT_{\mathrm{prob}}^{\mathrm{comp}}$ and are included by $\SAT_{\mathrm{causal}}^{\mathrm{lin}}$, they would all then be $\NP$-hard and inside $\NP$, and so would all be $\NP$-complete. It is known both that $\SAT_{\mathrm{prob}}^{\mathrm{lin}}$ is inside $\NP$ and that $\NP$ is closed under many-one $\NP$ reductions; by Proposition~\ref{NP-reduction}, this places $\SAT_{\mathrm{causal}}^{\mathrm{lin}}$ inside $\NP$, as desired.

We turn now to the second statement. By the same reasoning, it suffices to show that $\SAT_{\mathrm{prob}}^{\mathrm{cond}}$ is $\ETR$-hard and that $\SAT_{\mathrm{causal}}^{\mathrm{poly}}$ is inside $\ETR$. We claim that $\SAT_{\mathrm{prob}}^{\mathrm{poly}}$ is inside $\ETR$; $\ETR$ is closed under many-one $\NP$ reductions \citep{ten2013data}, so Proposition~\ref{NP-reduction} will place $\SAT_{\mathrm{causal}}^{\mathrm{poly}}$ in $\ETR$ immediately.

To show that $\SAT_{\mathrm{prob}}^{\mathrm{poly}}$ is inside $\ETR$, we slightly extend a proof by \citep{ibeling2020probabilistic} that the problem is in $\PSPACE$. Suppose that $\varphi \in \mathcal{L}_{\text{prob}}^{\text{poly}}$ is satisfied by some model $\prob$. Again using the fact that $\exists\mathbb{R}$ is closed under $\mathsf{NP}$-reductions, we will provide a reduction of $\varphi$ to a formula $\psi \in \mathsf{ETR}$. Let $E$ contain all $\epsilon$ such that $\probsymbol(\epsilon)$ appears in $\varphi$. Then consider the system of equations
\begin{align*}
    \sum_{\delta \in \Delta_\varphi} \probsymbol(\delta) &= 1\\
    \sum_{\substack{\delta \in \Delta_\varphi\\ \delta \models \epsilon}}\probsymbol(\delta) &= \prob(\epsilon) \text{ for } \epsilon \in E.
\end{align*}
The measure $\prob$ satisfies the above system, so by Lemma~\ref{small-model}, the above system is satisfied by some model $\prob_{\text{small}}$ assigning positive probability to a subset $\Delta^+ \subseteq \Delta_\varphi$ of size at most $|E| \leq |\varphi|$. Thus adding to $\varphi$ the constraint $\sum_{\delta \in \Delta^+} \probsymbol(\delta) = 1$ and replacing each $\probsymbol(\epsilon)$ appearing in $\varphi$ with $\sum_{\delta \in \Delta^+: \delta \models \epsilon}\probsymbol(\delta)$ gives a formula $\psi$ belonging to $\mathsf{ETR}$ which has a model, namely $\prob_{\text{small}}$---and conversely, the mutual unsatisfiability of the $\delta \in \Delta^+$, together with the fact that they sum to unity, ensures that any model of $\psi$ is a model of $\varphi$. Further, the size constraints on $E$ and $\Delta^+$ ensure that $\psi$ can be formed in polynomial time.

Let us conclude the proof by showing that $\SAT_{\mathrm{prob}}^{\mathrm{cond}}$ is $\exists\mathbb{R}$-hard. To do this, consider an $\exists \mathbb{R}$-inverse problem instance $\varphi$ with variables $x_1,...,x_n$. It suffices to find in polynomial time a $\SAT_{\mathrm{prob}}^{\mathrm{cond}}$ instance $\psi$ preserving and reflecting satisfiability. We first describe the reduction and then show that it preserves and reflects satisfiability.
    
    Corresponding to the variables $x_1,...,x_n$, define fresh events $\delta_1,...,\delta_{n} \in \sigma(\mathrm{Prop})$. Define fresh, disjoint events $\delta_1^\prime,...,\delta_n^\prime$. Let $\psi$ be the conjunction of the constraints
\begin{align*}
    \frac{1}{n} \geq \mathbf{P}(\delta_i) &\geq  \frac{1}{4n} & \text{ for }i=1,...,n\\
    \mathbf{P}(\delta_i | \delta_j) = \textbf{P}(\delta_i) \land \mathbf{P}(\delta_i \land \delta_j) &= \frac{1}{4n^2} & \text{for } x_i \cdot x_j = 1 \text{ in } \varphi\\
    \textbf{P}(\delta_i^\prime) = \textbf{P}(\delta_i)\land \textbf{P}(\delta_j^\prime) = \textbf{P}(\delta_j)\land\textbf{P}(\delta_i^\prime \lor \delta_j^\prime) &= \mathbf{P}(\delta_k) &\text{for } x_i + x_j = x_k \text{ in } \varphi.
\end{align*}

The formula $\psi$ is not yet in $\mathcal{L}_{\text{cond}}^{\text{prob}
}$, since it contains constants of the form $\nicefrac{1}{N}$. Replace each constant $\nicefrac{1}{N}$ with $\mathbf{P}(\epsilon_N)$, requiring that the fresh events $\epsilon_1,...,\epsilon_N$ are disjoint with $\textbf{P}(\lor_i \epsilon_i)=1$ and $\textbf{P}(\epsilon_i) = \textbf{P}(\epsilon_j)$ for $i =1,...,N$.

This completes our description of the reduction. The map $x_i \mapsto x_i/2n$ sends satisfying solutions of $\varphi$ to those of $\psi$, and the inverse map $\mathbb{P}(\delta_i) \mapsto \mathbb{P}(\delta_i) \cdot 2n$ sends satisfying solutions of $\psi$ to those of $\varphi$. Further, the operations performed are simple, and the introduced events $\delta_i, \delta_i^\prime, \epsilon_i$ and the constraints containing them are short, so the reduction is polynomial-time.
\end{proof}

\section{Conclusion and Outlook} \label{section:conclusion}

We have shown that questions posed in probabilistic \emph{causal} languages can be systematically reduced to purely probabilistic queries, showing that the former are---from a computational perspective---no more complex than the latter. In fact, we demonstrated a kind of bifurcation between two classes of languages. On the one hand, languages encompassing at most \emph{addition} enjoy an $\NP$-complete satisfiability problem, whether the language is causal or not. However, as soon as we admit even a modicum of \emph{multiplication} into the language, causal and probabilistic languages become hard for the class $\ETR$, and even the full language of polynomials over (causal) probability terms is $\ETR$-complete. At the low end, this applies to a language with no explicit addition or multiplication, but just inequalities between conditional probability terms, or even simple independence statements for pairs of variables. As clarified in the resulting landscape of formal systems, we have identified an important sense in which causal reasoning is no more difficult than pure probabilistic reasoning. The substantial empirical and expressive gulf between causation and ``mere (statistical) association’’ is evidently not reflected in a complexity gap.

It should be acknowledged that, from the standpoint of inferential practice, questions of the form (\ref{mainequation}) constitute just one part of a larger methodological  pipeline. In some sense this is only a final stage in the process of going from an inductive problem to a deductive conclusion. The formulation of reasonable inductive assumptions can itself be an arduous task, as can translating those assumptions into a language like $\mathcal{L}_{\mathrm{prob}}$ or $\mathcal{L}_{\mathrm{causal}}$ (that is, into the set $\Gamma$). Take once again the example of do-calculus (Example \ref{ex:do}). The idea behind this method is that in many contexts investigators will be in a position to make reasonable \emph{qualitative} (viz. graphical) assumptions, perhaps justified by expert knowledge, to the effect that some variables are \emph{not} causally impacted in a direct way by certain other variables. Even when this method involves nothing more than assuming a specific causal (directed acyclic) graph, it may still take work to determine which causal-probabilistic statements are licensed by the graph. Many subtasks in this connection have been studied. For instance, determining whether three sets of variables in a graph stand in the so called \emph{d-separation} relation (which in turn guarantees conditional independence) is known to be very easy (it is \emph{linear time}; see, e.g., \citealt{Schachter}). Nonetheless, there are certainly other questions related to complexity that one might ask in this and other settings.

Moving beyond statistical and causal inference tasks narrowly construed, the results in this article raise a number of further research questions, both technical and conceptual. For instance, one can easily imagine versions of our causal languages in a multi-agent setting, with a (causal) probability operator $\probsymbol_a$ for multiple agents $a$. As has been widely recognized, strategic interaction routinely involves reasoning about causality and counterfactuals (see, e.g., \citealt{Stalnaker}). Existing formal proposals for capturing these styles of reasoning have been largely qualitative, with counterfactual patterns formalized using models of belief revision rather than structural causal models (see, e.g., \citealt{Board}). Whereas (``pure'') probability-logical languages have been thoroughly explored in the game theory literature (e.g., \citealt{HeifetzMongin}), the causal-probability-logical languages studied here would be quite natural to investigate in that context. 
Echoing our themes in the present article, what happens to computational complexity in this multi-agent setting, and specifically would a reduction to pure (multi-)probability would still be possible? 

In a more technical vein, there are natural questions about further extensions to even the most expressive languages we considered. To take just one example, much of probabilistic and causal reasoning employs tools from information theory like (conditional) entropy that in turn rely on logarithmic principles, or alternatively (via inversion), reasoning about \emph{exponentiation}. A major open problem in logic---known as Tarski's exponential function problem---is to determine whether the first-order theory of the reals with exponentiation is decidable. Short of that, one might hope to show that some of the weaker (causal-)probability languages studied here remain decidable, perhaps even of relatively low complexity, when exponentiation is added. However, for the strongest languages, such as $\mathcal{L}^\mathrm{poly}_{\mathrm{prob}}$, this may prove difficult. As \cite{Wilkie} have shown, decidability of the \emph{existential} theory of the reals with the unary function $e^x$ would already imply a positive answer to Tarski's problem. 

The reader will surely think of further questions and extensions pertaining to our work in this article. We hope that the systems, results, and methods offered here will be useful in these various directions moving forward, and more generally will help to catalyze further research at the fruitful intersection of logic, probability, causality, and complexity.

\bibliographystyle{apalike}
\bibliography{refs}{}

\newcommand{\SortNoop}[1]{}
\begin{thebibliography}{}

\bibitem[Abadi and Halpern, 1994]{Abadi}
Abadi, M. and Halpern, J.~Y. (1994).
\newblock Decidability and expressiveness for first-order logics of
  probability.
\newblock {\em Information and Computation}, 112:1--36.

\bibitem[Abrahamsen et~al., 2018]{abrahamsen2018art}
Abrahamsen, M., Adamaszek, A., and Miltzow, T. (2018).
\newblock The art gallery problem is $\exists\mathbb{R}$-complete.
\newblock In {\em Proceedings of the 50th Annual ACM SIGACT Symposium on Theory
  of Computing}, pages 65--73.

\bibitem[Abrahamsen et~al., 2021]{abrahamsen2021training}
Abrahamsen, M., Kleist, L., and Miltzow, T. (2021).
\newblock Training neural networks is $\exists\mathbb{R}$-complete.
\newblock In {\em Proceedings of the Thirty-fifth Conference on Neural
  Information Processing Systems (NeurIPS)}.

\bibitem[Aleksandrowicz et~al., 2017]{Aleksandrowicz}
Aleksandrowicz, G., Chockler, H., Halpern, J.~Y., and Ivrii, A. (2017).
\newblock The computational complexity of structure-based causality.
\newblock {\em Journal of Artificial Intelligence Research}, 58:431--451.

\bibitem[Bareinboim et~al., 2022]{bareinboim2020pearl}
Bareinboim, E., Correa, J., Ibeling, D., and Icard, T. (2022).
\newblock On {P}earl's hierarchy and the foundations of causal inference.
\newblock In Geffner, H., Dechter, R., and Halpern, J.~Y., editors, {\em
  Probabilistic and Causal Inference: The Works of Judea Pearl}, pages
  509--556. ACM Books.

\bibitem[Belot, 2020]{BELOT2020159}
Belot, G. (2020).
\newblock Absolutely no free lunches!
\newblock {\em Theoretical Computer Science}, 845:159--180.

\bibitem[Bil{\`o} and Mavronicolas, 2017]{bilo2017existential}
Bil{\`o}, V. and Mavronicolas, M. (2017).
\newblock $\exists \mathbb{R}$-complete decision problems about symmetric
  {N}ash equilibria in symmetric multi-player games.
\newblock In {\em 34th Symposium on Theoretical Aspects of Computer Science
  (STACS 2017)}. Schloss Dagstuhl-Leibniz-Zentrum f\"{u}r Informatik.

\bibitem[Board, 2004]{Board}
Board, O. (2004).
\newblock Dynamic interactive epistemology.
\newblock {\em Games and Economic Behavior}, 49(1):49--80.

\bibitem[Canny, 1988]{canny1988some}
Canny, J. (1988).
\newblock Some algebraic and geometric computations in pspace.
\newblock In {\em Proceedings of the twentieth annual ACM symposium on Theory
  of computing}, pages 460--467.

\bibitem[Cardinal, 2015]{cardinal2015computational}
Cardinal, J. (2015).
\newblock Computational geometry column 62.
\newblock {\em ACM SIGACT News}, 46(4):69--78.

\bibitem[{\SortNoop{Cate}}ten~Cate et~al., 2013]{ten2013data}
{\SortNoop{Cate}}ten~Cate, B., Kolaitis, P.~G., and Othman, W. (2013).
\newblock Data exchange with arithmetic operations.
\newblock In {\em Proceedings of the 16th International Conference on Extending
  Database Technology}, pages 537--548.

\bibitem[Cook, 1971]{cook1971complexity}
Cook, S.~A. (1971).
\newblock The complexity of theorem-proving procedures.
\newblock In {\em Proceedings of the third annual ACM symposium on Theory of
  computing}, pages 151--158.

\bibitem[Darwiche, 2021]{darwiche2022causal}
Darwiche, A. (2021).
\newblock Causal inference using tractable circuits.
\newblock In {\em Proceedings of the Thirty-fifth Conference on Neural
  Information Processing Systems (NeurIPS)}.

\bibitem[Duarte et~al., 2021]{Duarte}
Duarte, G., Finkelstein, N., Knox, D., Mummolo, J., and Shpitser, I. (2021).
\newblock An automated approach to causal inference in discrete settings.
\newblock {\em arXiv preprint arXiv:2109.13471v1}.

\bibitem[Efron, 1978]{Efron}
Efron, B. (1978).
\newblock Controversies in the foundations of statistics.
\newblock {\em The American Mathematical Monthly}, 85(4):231--246.

\bibitem[Eiter and Lukasiewicz, 2002]{Eiter}
Eiter, T. and Lukasiewicz, T. (2002).
\newblock Complexity results for structure-based causality.
\newblock {\em Artificial Intelligence}, 142:53--89.

\bibitem[Erickson et~al., 2020]{erickson2020smoothing}
Erickson, J., Van Der~Hoog, I., and Miltzow, T. (2020).
\newblock Smoothing the gap between {NP} and $\exists\mathbb{R}$.
\newblock In {\em 2020 IEEE 61st Annual Symposium on Foundations of Computer
  Science (FOCS)}, pages 1022--1033. IEEE.

\bibitem[Fagin et~al., 1990]{fagin1990logic}
Fagin, R., Halpern, J.~Y., and Megiddo, N. (1990).
\newblock A logic for reasoning about probabilities.
\newblock {\em Information and computation}, 87(1-2):78--128.

\bibitem[{\SortNoop{Finetti}}de~Finetti, 1937]{Finetti1937}
{\SortNoop{Finetti}}de~Finetti, B. (1937).
\newblock La pr\'{e}vision: ses lois logiques, ses sources subjectives.
\newblock {\em Annales de l'Institut Henri Poincar\'{e}}, 7:1--68.

\bibitem[Gelman and Shalizi, 2013]{GelmanShalizi}
Gelman, A. and Shalizi, C.~R. (2013).
\newblock Philosophy and the practice of {B}ayesian statistics.
\newblock {\em British Journal of Mathematical and Statistical Psychology},
  66:8--38.

\bibitem[Halpern and Vardi, 1989]{HalpernVardi}
Halpern, J. and Vardi, M. (1989).
\newblock The complexity of reasoning about knowledge and time.
\newblock {\em Journal of Computer and System Sciences}, 38:195--237.

\bibitem[Halpern, 2000]{Halpern2000}
Halpern, J.~Y. (2000).
\newblock Axiomatizing causal reasoning.
\newblock {\em Journal of Artificial Intelligence Research}, 12:317--337.

\bibitem[Heifetz and Mongin, 2001]{HeifetzMongin}
Heifetz, A. and Mongin, P. (2001).
\newblock Probability logic for type spaces.
\newblock {\em Games and Economic Behavior}, 35:31--53.

\bibitem[Ibeling, 2018]{ibeling18}
Ibeling, D. (2018).
\newblock Causal modeling with probabilistic simulation models.
\newblock In {\em Proceedings of the 5th International Workshop on
  Probabilistic Logic Programming (PLP)}, pages 36--48.

\bibitem[Ibeling and Icard, 2020]{ibeling2020probabilistic}
Ibeling, D. and Icard, T. (2020).
\newblock Probabilistic reasoning across the causal hierarchy.
\newblock In {\em Proceedings of the 34th AAAI Conference on Artificial
  Intelligence; revised as arXiv:2001.02889v5}.

\bibitem[Ibeling and Icard, 2021]{II21}
Ibeling, D. and Icard, T. (2021).
\newblock A topological perspective on causal inference.
\newblock In {\em Proceedings of the Thirty-fifth Conference on Neural
  Information Processing Systems (NeurIPS)}.

\bibitem[Ibeling et~al., 2022]{IIMM22}
Ibeling, D., Icard, T., Mierzewski, K., and Moss\'{e}, M. (2022).
\newblock Probing the quantitative--qualitative divide in probabilistic
  reasoning.
\newblock Unpublished Manscript.

\bibitem[Kurucz, 2007]{Kurucz}
Kurucz, A. (2007).
\newblock Combining modal logics.
\newblock In van Benthem, J., Blackburn, P., and Wolter, F., editors, {\em
  Handbook of Modal Logic}, pages 869--924. Elsevier.

\bibitem[Luce, 1968]{luce1968numerical}
Luce, R.~D. (1968).
\newblock On the numerical representation of qualitative conditional
  probability.
\newblock {\em The Annals of Mathematical Statistics}, 39(2):481--491.

\bibitem[Macintyre and Wilkie, 1995]{Wilkie}
Macintyre, A. and Wilkie, A.~J. (1995).
\newblock On the decidability of the real exponential field.
\newblock In Odifreddi, P., editor, {\em Kreiseliana. About and Around {G}eorg
  {K}reisel}, pages 441--467. A. K. Peters.

\bibitem[Neyman, 1977]{Neyman}
Neyman, J. (1977).
\newblock Frequentist probability and frequentist statistics.
\newblock {\em Synthese}, 36(1):97--131.

\bibitem[Ognjanovi\'{c} et~al., 2016]{PL}
Ognjanovi\'{c}, Z., Ra\v{s}kovi\'{c}, M., and Markovi\'{c}, Z. (2016).
\newblock {\em Probability Logics}.
\newblock Springer.

\bibitem[Pearl, 1995]{pearl1995}
Pearl, J. (1995).
\newblock Causal diagrams for empirical research.
\newblock {\em Biometrika}, 82(4):669--710.

\bibitem[Pearl, 2009]{Pearl2009}
Pearl, J. (2009).
\newblock {\em Causality}.
\newblock Cambridge University Press.

\bibitem[Roth, 1996]{ROTH1996273}
Roth, D. (1996).
\newblock On the hardness of approximate reasoning.
\newblock {\em Artificial Intelligence}, 82(1):273--302.

\bibitem[Ruiz-Vanoye et~al., 2011]{ruiz2011survey}
Ruiz-Vanoye, J.~A., P{\'e}rez-Ortega, J., D{\'\i}az-Parra, O.,
  Frausto-Sol{\'\i}s, J., Huacuja, H. J.~F., Cruz-Reyes, L., et~al. (2011).
\newblock Survey of polynomial transformations between {NP}-complete problems.
\newblock {\em Journal of computational and applied mathematics},
  235(16):4851--4865.

\bibitem[Schachter, 1988]{Schachter}
Schachter, R.~D. (1988).
\newblock Probabilistic inference and influence diagrams.
\newblock {\em Operations Research}, 36:589--605.

\bibitem[Schaefer, 2009]{schaefer2009complexity}
Schaefer, M. (2009).
\newblock Complexity of some geometric and topological problems.
\newblock In {\em International Symposium on Graph Drawing}, pages 334--344.
  Springer.

\bibitem[Schaefer, 2013]{schaefer2013realizability}
Schaefer, M. (2013).
\newblock Realizability of graphs and linkages.
\newblock In {\em Thirty Essays on Geometric Graph Theory}, pages 461--482.
  Springer.

\bibitem[Scott and Krauss, 1966]{Scott1966AssigningPT}
Scott, D. and Krauss, P. (1966).
\newblock Assigning probabilities to logical formulas.
\newblock {\em Studies in Logic and the Foundations of Mathematics},
  43:219--264.

\bibitem[Shalev-Shwartz and Ben-David, 2014]{MLbook}
Shalev-Shwartz, S. and Ben-David, S. (2014).
\newblock {\em Understanding Machine Learning: From Theory to Algorithms}.
\newblock Cambridge University Press.

\bibitem[Speranski, 2017]{Speranski}
Speranski, S.~O. (2017).
\newblock Quantifying over events in probability logic: An introduction.
\newblock {\em Mathematical Structures in Computer Science}, 27(8):1581--1600.

\bibitem[Spirtes et~al., 2000]{spirtes2000causation}
Spirtes, P., Glymour, C.~N., and Scheines, R. (2000).
\newblock {\em Causation, Prediction, and Search}.
\newblock The MIT Press.

\bibitem[Stalnaker, 1996]{Stalnaker}
Stalnaker, R.~C. (1996).
\newblock Knowledge, belief and counterfactual reasoning in games.
\newblock {\em Economics and Philosophy}, 12:133--163.

\bibitem[Suppes and Zanotti, 1981]{suppes:zan81}
Suppes, P. and Zanotti, M. (1981).
\newblock {When are probabilistic explanations possible?}
\newblock {\em Synthese}, 48:191--199.

\end{thebibliography}

\end{document}